\documentclass{article}

\usepackage{PRIMEarxiv}

\usepackage[utf8]{inputenc} 
\usepackage[T1]{fontenc}    
\usepackage{hyperref}       
\usepackage{url}            
\usepackage{booktabs}       
\usepackage{amsfonts}       
\usepackage{nicefrac}       
\usepackage{microtype}      
\usepackage{lipsum}
\usepackage{fancyhdr}       
\usepackage{graphicx}       
\usepackage{amsmath}
\usepackage{graphicx} 
\usepackage{rotating}
\usepackage{subfigure}

\usepackage{amsthm}
\usepackage{newtxtext} %
\usepackage{newtxmath}
\graphicspath{{media/}}     

\newtheorem {proposition}{Proposition}[section]

\newtheorem {theorem}[proposition]{Theorem}

\newtheorem {remark}[proposition]{Remark}
\newtheorem {example}{Example}[section]
\newtheorem {definition}{Definition}[section]

\title{Fairness constraint in Structural Econometrics and Application to fair estimation using Instrumental Variables}

\author{
  Samuele Centorrino  \\
  Stony Brook University \\
  Stony Brook \\
  NY, USA \\
  \texttt{ samuele.centorrino@stonybrook.edul} \\
   \And
  Jean-Pierre Florens \\
  Toulouse School of Economics\\
  University of Toulouse Capitole \\
  Toulouse, France\\
  \texttt{jean-pierre.florens@tse-fr.eu} \\
    \And
  Jean-Michel Loubes  \\
   Institut de Math\'ematiques de Toulouse,\\
   Universit\'e Toulouse Paul Sabatier, \\
  Toulouse, France\\
  \texttt{loubes@math.univ-toulouse.fr} \\
}

 \allowdisplaybreaks

\newcommand{\IR}{\mathbb{R}}

\newcommand{\indep}{\perp \!\!\! \perp}

\begin{document}
\maketitle

\begin{abstract}
A supervised machine learning algorithm determines a model from a learning sample that will be used to predict new observations. To this end, it aggregates individual characteristics of the observations of the learning sample. But this information aggregation does not consider any potential selection on unobservables and any status-quo biases which may be contained in the training sample. The latter bias has raised concerns around the so-called \textit{fairness} of machine learning algorithms, especially towards disadvantaged groups. In this chapter, we review the issue of fairness in machine learning through the lenses of structural econometrics models in which the unknown index is the solution of a functional equation and issues of endogeneity are explicitly accounted for. We model fairness as a linear operator whose null space contains the set of strictly {\it fair} indexes. A  {\it fair} solution is obtained by projecting the unconstrained index into the null space of this operator or by directly finding the closest solution of the functional equation into this null space. We also acknowledge that policymakers may incur a cost when moving away from the status quo. Achieving  \textit{approximate fairness} is obtained by introducing a fairness penalty in the learning procedure and balancing  more or less heavily the influence between the status- quo and a full fair solution.
\end{abstract}


\section{Introduction}

Fairness has been a growing field of research in Machine Learning, Statistics, and Economics over the recent years. The purpose of such work is to monitor data driven models that rely too much on correlations with a variable which should not be used in the data. In particular for complex machine learning methods, the outcome of the algorithm can be considered as a black-box which provides a prediction without being able to understand the reasons for it. Accuracy of the model when forecasting has become the gold standard. Yet in many cases, the decisions are taken at the expenses of minority groups or driven by some characteristics of the observations from the learning sample that appear to be confounding variables. The model fitted by the algorithm may rely on correlations with a variable whose use is irrelevant. This variable is a potential source of bias which influences the behaviour of the algorithm. In many situations, the choice of this variable, known as the sensitive variable, can be driven by ethical issues, legal issues or regulation issues. From a moral point of view penalizing a group of individuals is an unfair decision. From a legal perspective \footnote{Artificial Intelligence European Act 2021 }, unfair algorithmic decisions are prohibited for a large number of applications, including access to education, welfare system or microfinance. To comply with fairness regulations, the institution may either choose to change the decision process to remove biases using affirmative actions or try to base their decision on a fair version of the outcome. \vskip .1in
A typical example is given by algorithmic decisions of machine learning procedures. When bias is present in the learning sample, the algorithm's output can be different for different subgroups of populations, while regulations may impose that such groups ought to be treated in the same way. For instance, discrimination can occur on the basis of gender or ethnic origin. A typical example is the one of automatic human resources (HR) decisions that are often influenced by gender. In available databases, men and women may self-select in some job categories due to past or present preferences or cultural customs. Some jobs are considered as male jobs while other jobs are female dominant. In such unbalanced datasets, the machine learning procedure learns that the gender matters and thus transform the correlation into a causality by using the gender variable as a causal variable in the future decisions. From a legal point of view this biased decision leads to punishable gender discrimination. We refer to \cite{de2019bias} for more insights on this gender gap. Disparate treatment for university admissions suffer from the same problems. We point out the well used dataset of law schools admissions described in \cite{mcintyre2018law}, which is used as a common benchmark to evaluate bias of algorithmic decisions.\vskip .1in
Imposing fairness is thus about mitigating this unwanted bias and preventing the sensitive variable to influence decisions. Fairness can be divided into two main categories. A first definition of fairness is to impose that the output of the algorithm is the same for all groups, hence that the sensitive variable does not play any role in the decision. Such equality of treatment is referred to as {\it statistical parity}. \\
 \indent A different fairness condition is given by the fact that we do not restrict to models giving the same forecast for the different subgroups but we rather wish to ensure that the algorithm has the same performance over all possible subgroups. For instance an algorithm could perform well for a category of the population but fail for others. It is the case with the well known predictive justice algorithm described in~\cite{angwin2016machine} where discrimination towards Afro-American is proven. When the performance of the algorithm is different for different groups of individuals, the notion of fairness which is violated is known as {\it equality of odds}. \vskip .1in Bias mitigation has been studied in this framework over the last years. Many methods have been developed to achieve fairness of algorithmic decisions. The proposed algorithm are usually divided into three categories. The first method is a post-processing method which consist in removing bias from the learning sample to learn a fair algorithm. The second way consists in imposing fairness constraint while learning the algorithm and balancing the desired fairness with the accuracy of the model. This method is an in-processing method. Finally, the last method is a post-processing method where the output of a possibly unfair algorithm is processed to achieve the desired level of fairness, modelled using different fairness measures. All three methodologies required a proper definition for fairness and a choice of fairness measures to quantify it.\vskip .1in
 \indent Achieving {\it full fairness} consists in removing completely the effect of the sensitive variable. it often involves an important changes with respect to the unfair case and comes at the expenses of accuracy of the algorithm, when the accuracy is measured using the biased distribution of the data set. When the loss of accuracy is considered too important by the designer of the model, an alternative consists in weakening the fairness constraint by choosing a way to quantify it. Unfortunately, there is not a universal measure to quantify a fair model since the notion of dependency are multiple. Complying some criterion at the same time is even proven to be impossible as pointed out in~\cite{3433949}. Hence the stakeholder has to choose a fairness criterion and then build a model for the which the fairness level will be above a certain chosen threshold. The model will thus be called {\it approximately fair}. We point out that choices of different fairness constraint give rise to different fair models. \vskip .1in
 \indent To sum up, fairness with respect to a given variable, $S$, is about controlling the influence of its distribution and preventing its influence on an estimator. We refer to \cite{barocas2016big}, \cite{chouldechova2017fair}, \cite{pmlr-v81-menon18a}, \cite{del2018obtaining}, \cite{Oneto2020} and \cite{besse2021survey} and references therein for deeper insights on the notion of bias mitigation and fairness. \vskip .1in \indent In the following we present the challenges of fairness constraints in econometrics. Some works have studied the importance of fairness in economics (see, for instance, \cite{rambachan2020economic}, \cite{lee2021formalising}, \cite{hoda2018moral}, \cite{hu2020fair}, \cite{kasy2021}, and references therein). As seen previously, the term fairness is polysemic and covers various notions. We will focus on the role and on the techniques that can be used to impose fairness in a specific class of econometrics models. \\
 \indent 
 
Let us consider the example in which an institution must make a decision concerning a group of individuals. For instance, this could be a university admitting new students based on their expected performance in a test; or a company deciding the hiring wage of new employees. This decision is made by an algorithm, which we suppose works in the following way. For a given vector of individual's characteristics, denoted by $X$, this algorithm computes a score $\varphi(X) \in \IR$, and makes a decision based on the value of this score, which is determined by a functional $\mathcal{D}$ of $\varphi$. We are not specific about the exact form of $\mathcal{D}(\varphi)$. For instance, this could be a threshold function in which students are admitted if the score is higher than or equal some values $C$, and they are not admitted otherwise. The algorithm is completed by a learning model, which is written as follows
\begin{equation} \label{mainmod}
Y = \varphi(X) + U,
\end{equation}
where $Y$ is the outcome and $U$ is a statistical error. For instance, $Y$ could be the test result from previous applicants. We let $X = (Z,S) \in \IR^{p+1}$ and $\mathcal{X} = \mathcal{Z} \times \mathcal{S}$ to be the support of the random vector $X$. We further restrict $\varphi \in L^2(X)$, with $L^2$ being the space of square integrable functions with respect to some probability distribution. This learning model is used to approximate the score, $\varphi(X)$, which is then used in the decision model. 

Let us assume that historical data show that students from private high schools obtain higher test scores than students in public high schools. The concern with fairness in this model is twofold. On the one hand, if the distinction between public and private school is used as a predictor, students from private schools will always have a higher probability of being admitted to a given university. On the other hand, the choice of school is an endogenous decision that is taken by the individual and may be determined by variables which are unobservable to the econometrician. Hence the bias will be reflected both in the lack of fairness in past decision-making processes and the endogeneity of individual's choices in the observational data. Hence, predictions and admission decisions may be unfair towards the minority class and bias the decision process, possibly leading to discrimination. To overcome this issue, we consider that decision makers can embed in their learning model a \textit{fairness} constraint. This fairness constraint limits the relationship between the score $\varphi(X)$ and $S$. Imposing a fairness constraint directly on $\varphi$ and not on $\mathcal{D}(\varphi)$ is done for technical convenience, as $\mathcal{D}(\varphi)$ is often nonlinear, which complicates substantially the estimation and prediction framework.

More generally, our aim is to study the consequences of incorporating a fairness constraint in the estimation procedure when the score, $\varphi$, solves a linear inverse problem of the type
\[
K\varphi = r,
\]
where $K$ is a linear operator. A leading example of this setting are nonparametric instrumental regressions \cite{newey2003,hall2005,darolles2011}, as mentioned above, but many other models, such as linear and non-linear parametric regressions and additive nonparametric regressions can fit in this general framework \cite{carrasco2007h}.

Let $\mathcal{E} = \lbrace \varphi \in L^2(X) \rbrace$, and $\mathcal{G}$ be the space of functions of $X$ which satisfy a fairness constraint. We model the latter as a linear operator $F : \mathcal{E} \rightarrow \mathcal{G} $ such that \begin{equation} \label{eq:fullfair} F \varphi = 0. \end{equation} 

That is, the kernel of the operator $F$ is the space of those functions which satisfy a fairness restriction, $\mathcal{N}(F) =\{ g \in \mathcal{E}, \: Fg=0 \} $. The \textit{full fairness} constraint implies to restrict the solutions to the functional problem to the kernel of the operator. To weaken this requirement, we also consider relaxations of the condition and define an \textit{approximate fairness} condition as $$ \|F \varphi \| \leq \rho$$ for some well chosen balance parameter $\rho \geq 0$. 

In this work, we consider fairness according to the following definitions. 

\begin{definition}[Statistical Parity] \label{def:fairness1}
The algorithm $\varphi$ maintains statistical parity if, for every $s \in \mathcal{S}$,
\[
E \left[ \varphi(Z,s) \vert S = s \right] = E \left[ \varphi(X) \right].
\]
\end{definition}

\begin{definition}[Irrelevance in prediction] \label{def:fairness2}
The algorithm $\varphi$ does not depend on $S$. \\ That is for all $ s \in \mathcal{S}$,
\[
\frac{\partial \varphi(x)}{\partial s} = 0.
\]
\end{definition}

The first definition implies that the function $\varphi$ is fair when individuals are treated the same, on average, irrespective of the value of the sensitive attribute, $S$. For instance, if $S$ is a binary characteristics of the population, with $S=1$ be the protected group, Definition \ref{def:fairness1} implies that the average score for the group $S=0$ and the average score for the group $S = 1$ are the same. Notice that this definition of fairness does not ensure that two individuals with the same vector of characteristics $Z = z$, but with different value of $S$ are treated in the same way. This is instead true for our second definition of fairness. In this case, fairness is defined as the lack of dependence of $\varphi$ on $S$, which implies the \textit{equality of odds} for individuals with the same vector of characteristics $Z = z$. We want to point out however that both these definitions may fail to deliver fairness if the correlation between $Z$ and $S$ is very strong. In our example above, if students going to private schools have higher income than students going to public schools, and income positively affects the potential score, then discrimination would still occur on the basis of income. \vskip .1in

Other definitions of fairness are possible. In particular, definitions that impose restriction on the entire distribution of $\varphi$ given $S$. These constraints are nonlinear and thus more cumbersome to deal with in practice, and we defer their study to future work. 

 \section{Examples in Econometrics}
 
We let $\mathcal{F}_1$ and $\mathcal{F}_2$ be the set of square integrable functions which satisfy definitions \ref{def:fairness1} and \ref{def:fairness2}, respectively. We consider below examples in which the function $\varphi_F$ satisfies
\[
\varphi_F = {{\rm arg}\min}_{f \in \mathcal{F}_j} \mathbb{E} \left[ \left( Y - f(X) \right)^2 \vert W =w\right],
\]
with $j = \lbrace 1,2 \rbrace$, and where $W$ is a vector of instrumental variables.
\vskip .1in

\subsection{Linear IV model} 
Consider the example of a linear model in which $\varphi(X) = Z^\prime \beta + S^\prime\gamma$, with $Z,\beta \in \IR^p$ and $S,\gamma \in \IR^q$. We take both $Z$ and $S$ to be potentially endogenous and we have a vector of instruments $W \in \IR^k$, such that $k \geq p+q$ and $E\left[ W^\prime U \right] = 0$. 

We let $X=(Z^\prime,S^\prime)^\prime$ be the vector of covariates, and $\varphi = (\beta^\prime,\gamma^\prime)^\prime$ be the vector of unknown coefficients. 

For simplicity of exposition, we maintain the assumption that the vector $$\begin{bmatrix} X\\W \end{bmatrix} \sim N \left( {\bf 0}_{p+q+k}, \begin{bmatrix} \Sigma_X & \Sigma^\prime_{XW}\\ \Sigma_{XW} &I_k\end{bmatrix}\right),$$ where ${\bf 0}_{p+q+k}$ is a vector of zeroes of dimension $p + q + k$, $I_k$ is the identity matrix of dimension $k$, and 
\[
\underbrace{\Sigma_X}_{(p + q) \times (p + q)} = \begin{bmatrix} \Sigma_Z & \Sigma^\prime_{ZS} \\ \Sigma_{ZS} & \Sigma_S \end{bmatrix}, \qquad \underbrace{\Sigma_{XW}}_{k \times (p + q)} = \begin{bmatrix} \Sigma_{ZW} & \Sigma_{SW} \end{bmatrix}.
\]
The unconstrained value of $\varphi$ is therefore given by
\[
\varphi = \left( \Sigma^\prime_{XW} \Sigma_{XW} \right)^{-1} \Sigma^\prime_{XW} E \left[ W Y\right] = \left( K^\ast K \right)^{-1} K^\ast r.
\]

Because of the assumption of joint normality, we have that $E\left[ Z \vert S \right] = \Pi S$, where $\Pi = \Sigma^{-1}_S \Sigma_{ZS}$ is a $p\times q$ matrix.

\subsection{A nonlinear IV model with a binary sensitive attribute}
Let $Z \in \IR^p$ be a continuous variable and $S = \lbrace 0, 1 \rbrace^q$ a binary random variable. For instance, $S$ can characterize gender, ethnicity, or a dummy for school choice (public vs private). Because of the binary nature of $S$
\[
\varphi(X) = \varphi_0(Z) + \varphi_1(Z) S.
\]
Definition \ref{def:fairness1} implies that we are looking for functions $\lbrace \varphi_0,\varphi_1\rbrace$ such that 
\[
E\left[ \varphi_0(Z) \vert S = 0 \right] = E\left[ \varphi_0(Z) + \varphi_1(Z) \vert S = 1 \right]. 
\]
That is
\[
E\left[ \varphi_1(Z) \vert S = 1 \right] = E\left[ \varphi_0(Z) \vert S = 0 \right] - E\left[ \varphi_0(Z) \vert S = 1 \right].
\]
Definition \ref{def:fairness2} instead simply implies that $\varphi_1 = 0$, almost surely. In particular, under the fairness restriction, $\varphi_0(Z) = E[ Y \vert Z]$. We develop this example in more detail in Section \ref{sec:illustration}.

\subsection{Fairness and structural econometrics}

In a more general fashion, supervised machine learning models are often about prediction of a conditional moment or a conditional probability. However, in many leading examples in structural econometrics, the score function, $\varphi$ does not correspond directly to a conditional distribution or a conditional moment of the distribution of the learning variable $Y$. Let $\Gamma$ be the probability distribution generating the data. Then the function $\varphi$ to be solution to the following equation
\[
A \left( \varphi, \Gamma \right) = 0.
\]
A leading example is the one of Neyman-Fisher-Cox-Rubin \textit{potential outcome} models, in which $X$ represents a treatment and, for $X = \xi$, we can write
\begin{equation} \label{eq:model}
Y_\xi = \varphi(\xi) + U_\xi.
\end{equation}
If $E\left[ U_\xi \vert W \right] = 0$, this model leads to the nonparametric instrumental regression model mentioned above, in which the function $A (\varphi,\Gamma) = E \left[ Y - \varphi(X) \vert W \right] =0$, and the fairness condition is imposed directly on the function $\varphi$. This potential outcome model can however lead to other objects of interest. For instance, if we assume for simplicity that $(X,W) \in \IR^2$, and under a different set of identification assumptions, it can be proven that
\[
A (\varphi,\Gamma) = E \left[\frac{ d\varphi(X)}{d X} \vert W \right] -\frac{\frac{d E\left[ Y \vert W \right]}{d W}}{\frac{d E\left[ Z \vert W \right]}{d W}}= 0,
\]
which is a linear equation in $\varphi$ which combines integral and differential operators \cite{florens2008}. In this case, the natural object of interest is the first derivative of $\varphi(x)$, which is the marginal treatment effect. The fairness constraint is therefore naturally imposed on $\frac{d\varphi(x)}{dx}$. 

Another class of structural models which is not explicitly considered in this work is the class of nonlinear nonseparable models. In these models, we have that
\[
Y = \varphi(X,U), \text{ with } U \indep W \text{ and } U \sim \mathcal{U}[0,1],
\]
and $\varphi(\xi,\cdot)$ monotone increasing in its second argument. In this case, $\varphi$ is the solution of the following non-linear inverse problem
\[
\int P \left( Y \leq \varphi(x,u) \vert X = x, W = w \right) f_{X\vert W} (x\vert w) dx = u.
\]
The additional difficulty lays on how to impose a \textit{distributional} fairness constraint in this setting. We defer the treatment of this case to future research.


\section{Fairness for Inverse Problems }
Recall that the nonparametric instrumental regression (NPIV) model amounts to solving an inverse problem defined as follows. Consider $W$ the instrument, the NPIV regression model can be written as
\[ E(Y|W=w)= E(\varphi(Z,S)|W=w) \] We let $X = (Z,S) \in R^{p + q}$ and $\mathcal{X} = \mathcal{Z} \times \mathcal{S}$ to be the support of the random vector $X$. We further restrict $\varphi \in L^2(X)$, with $L^2$ being the space of square integrable functions with respect to some distribution $\mathbf{P}$. \\

If we let $r =E(Y|W=w)$ and $K \varphi =  E(\varphi(Z,S)|W=w) $, where $K$ is  conditional expectation operator, then the NPIV framework amounts to solving an inverse problem. That is, estimating a function $\varphi_\dagger \in \mathcal{E}$ defined as the solution of 
\begin{equation} \label{eq:inv}
r=K \varphi_\dagger.
\end{equation}
If the operator $K^*K$ is invertible, the solution of \eqref{eq:inv} is given by
\begin{equation} \label{est0}
\varphi_\dagger= (K^*K)^{-1}K^*r.
\end{equation}
The ill-posedness of the inverse problem in \eqref{est0} comes from the fact that, when the distribution of $(X,W)$ is continuous, the eigenvalues of the operator $K^*K$ have zero an as accumulation point. To prevent ill-posedness of the operator, a usual solution consists in using a regularisation techique \cite{engl1996regularization},  and references therein. In this paper, we use the so-called Tikhonov regularization, which imposes an $L^2$-penalty on the function $\varphi$ \cite{natterer1984error}. The regularized solution, as presented in \cite{engl1996regularization}, is $\varphi_\alpha$, defined as the solution of a penalized optimization program \[ \varphi_\alpha = {\rm arg}\min_{\varphi \in \mathcal{E}} \|r-K \varphi\|^2 + \alpha \| \varphi \|^2 \] with the solution written as 
\begin{equation} \label{tiko} \varphi_\alpha = (\alpha {\rm Id} + K^*K)^{-1}K^* r =R_\alpha(K) K^* r \end{equation}
where $R_\alpha(K)=(\alpha {\rm Id} + K^*K)^{-1}$ is a Tikhonov regularized operator. \vskip .1in

We consider the estimation of the function $\varphi$ from the following noisy observational model 
\begin{equation} \label{eq:model}
\hat{r}= K \varphi_\dagger + U_n, 
\end{equation}
where $U_n$ is an unknown random function with bounded norm $\|U_n\|^2=O(\delta_n)$ for a given sequence $\delta_n$ which tends to 0 when $n$ goes to infinity. The operator $K$ is taken to be known for simplicity. This estimation problem has been widely studied in the econometrics literature, and we provide details on the estimation of the operator in Section \ref{sec:estimation}. We refer for instance to \cite{darolles2011} for the asymptotic properties of the NPIV estimator when the operator $K$ is estimated from data. \vskip .1in 
We assume the following conditions 
\begin{itemize}
\item{\bf [A1]} $r \in \mathcal{R}(K)$ where $\mathcal{R}(K)$ stands for the range of the operator $K$
\item{\bf [A2]} The operator $K^*K$ is a one to one operator. This condition ensures the identifiability of $\varphi_0$.
\item{\bf [A3]} Source Condition : we assume that there exists $\beta \leq 2$ such that $$ \varphi_\dagger \in \mathcal{R}(K^*K)^{\frac{\beta}{2}}.$$ This condition relates the smoothness of the solution of equation \eqref{eq:inv} to the decay of the eigenvalues of the SVD decomposition of the operator $K$. It is well used in inverse problems, we refer to \cite{loubes2009review} for a review of the different smoothness conditions for inverse problems. In particular it guarantees that the Tikhonov regularized solution $\varphi_\alpha$ converges to the true solution $\varphi_\dagger$ at a rate of convergence given by 
\[ \| \varphi_\alpha- \varphi_\dagger \|^2 =O(\alpha^\beta).\] 
\end{itemize}

\section{Full fairness IV approximation} \label{sec:fullfair}
 In this model, full fairness of a function $\psi \in \mathcal{E}$ is achieved when $F\psi =0$, i.e when the function belongs to the Kernel of the fairness operator. Hence imposing fairness amounts to considering function that belong to the Kernel space $\mathcal{N}(F)$ and that are approximate solution of the function equation \eqref{eq:inv}.
 The Full Fairness condition may be seen as a very restrictive way to impose fairness. Actually, if the functional equation does not have a solution in $\mathcal{N}(F)$, full fairness will induce a loss of accuracy which is the so-called {\it price for fairness}. The projection to fairness has been studied in the regression framework in \cite{le2020projection}, \cite{chzhen2020fair} and \cite{jiang2020wasserstein}, for the classification task. \\
 Actually full fairness condition can be achieved in two different ways : either by looking at the solution of the inverse problem and then imposing a fair condition on the solution, or solving the inverse problem under the restriction that the solution is fair. We prove that the two procedures are not equivalent and lead to different estimation having different properties.
 \begin{center}
\includegraphics{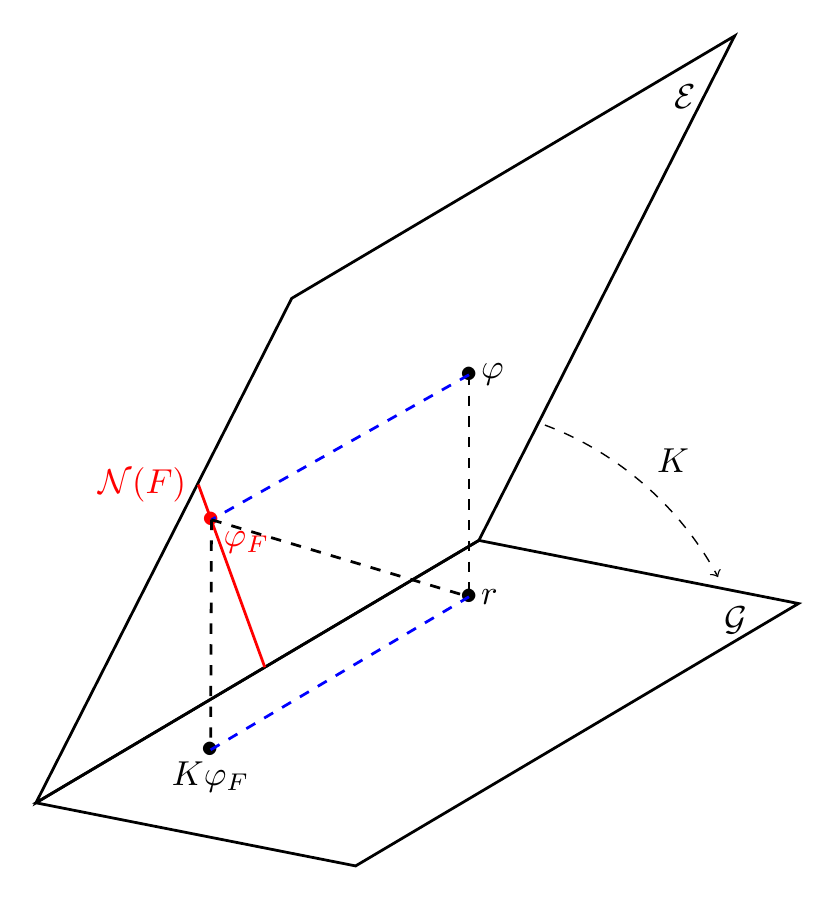} \label{fig:jolie}
\end{center}

Figure~\ref{fig:jolie} illustrates the situation where either the solution can be solved and then the fairness condition can be imposed or the solution is directly approximated in the set of fair functions.
 \subsection{Projection onto Fairness}
 \indent The first way consists in first considering the regularized solution to the inverse problem $\hat{\varphi}_\alpha$ defined as the Tikhonov regularized solution of the inverse problem
 $$ \hat{\varphi}_\alpha = {\rm arg}\min_{\varphi \in \mathcal{E}} \left( \| \hat{r}-K \varphi \|^2 +\alpha \| \varphi \|^2 \right) $$ which can be computed as
 $$ \hat{\varphi}_\alpha = (\alpha {\rm Id} + K^*K)^{-1}K^* \hat{r} =R_\alpha(K) K^* \hat{r} .$$
 Then the fair solution is defined as the projection onto the set which models the fairness condition $\mathcal{N}$, 
 $$ \hat{\varphi}_{\alpha,F} = {\rm arg}\min_{\varphi \in \mathcal{N}(F)} \| \hat{\varphi}_\alpha - \varphi \|^2$$
 In this framework, denote by $P: \mathcal{E} \rightarrow \mathcal{N}(F)$ the projection operator onto the kernel of the fairness operator. Hence we have 
 $$ \hat{\varphi}_{\alpha,F}=P \hat{\varphi}_\alpha.$$
 
 \begin{example}[Linear Model, continued.]
The constraint of statistical parity in Definition \ref{def:fairness1} implies that $$S^\prime(\Pi \beta + \gamma) = 0,$$ which is true as long as $\Pi \beta + \gamma = {\bf 0}_q$. Thus, we have that 
\[
\underbrace{F}_{q \times (p + q)} = \begin{bmatrix} \Pi & I_q \end{bmatrix},
\]
and 
\[
P = I_{p+q} - F^\prime \left( F F^\prime\right)^{-1} F = I_{p+q} - \begin{bmatrix} \Pi^\prime \left( I_q + \Pi \Pi^\prime\right)^{-1} \Pi & \Pi^\prime \left( I_q + \Pi \Pi^\prime\right)^{-1} \\ \left( I_q + \Pi \Pi^\prime\right)^{-1} \Pi & \left( I_q + \Pi \Pi^\prime\right)^{-1}\end{bmatrix},
\]
which immediately gives $FP = {\bf 0}_q$. Hence, the value of $\varphi_F = P\varphi$ is the projection of the vector $\varphi$ onto the null space of $F$. 

In the case of definition \ref{def:fairness2}, the fairness constraint is simply given by $\gamma = 0$. Let $$M_{ZW} = I_k - \Sigma_{ZW} \left( \Sigma^\prime_{ZW} \Sigma_{ZW} \right)^{-1} \Sigma^\prime_{ZW},$$ and $$A_{ZS} = \left( \Sigma^\prime_{ZW} \Sigma_{ZW} \right)^{-1} \Sigma^\prime_{ZW} \Sigma_{SW}.$$ When one wants to project the unconstrained estimator onto the constrained space, by the block matrix inversion lemma, we notice that
\begin{align*}
\varphi =& \begin{bmatrix} \left( \Sigma^\prime_{ZW} \Sigma_{ZW} \right)^{-1} + A_{ZS} \left( \Sigma^\prime_{SW} M_{ZW} \Sigma_{SW} \right)^{-1} A^\prime_{ZS} & - A_{ZS} \left( \Sigma^\prime_{SW} M_{ZW} \Sigma_{SW}\right)^{-1} \\ -\left( \Sigma^\prime_{SW} M_{ZW} \Sigma_{SW}\right)^{-1} A^\prime_{ZS} & \left( \Sigma^\prime_{SW} M_{ZW} \Sigma_{SW}\right)^{-1} \end{bmatrix} \begin{bmatrix} \Sigma^\prime_{ZW} E\left[ W Y \right] \\ \Sigma^\prime_{SW} E\left[ W Y \right]\end{bmatrix}\\
=& \begin{bmatrix} \left( \Sigma^\prime_{ZW} \Sigma_{ZW} \right)^{-1} \Sigma^\prime_{ZW} E\left[ W Y \right] - A_{ZS} \left( \Sigma^\prime_{SW} M_{ZW} \Sigma_{SW} \right)^{-1} \left( \Sigma^\prime_{SW} E\left[ W Y \right] - A^\prime_{ZS} \Sigma^\prime_{ZW} E\left[ W Y \right] \right) \\ \left( \Sigma^\prime_{SW} M_{ZW} \Sigma_{SW}\right)^{-1}\left( \Sigma^\prime_{SW} E\left[ W Y \right] - A^\prime_{ZS} \Sigma^\prime_{ZW} E\left[ W Y \right] \right)\end{bmatrix} \\
=& \begin{bmatrix} \left( \Sigma^\prime_{ZW} \Sigma_{ZW} \right)^{-1} \Sigma^\prime_{ZW} E\left[ W Y \right] - A_{ZS} \gamma \\ \gamma \end{bmatrix}.
\end{align*}
Therefore, we have that
\[
\varphi_F = P \varphi =\begin{bmatrix} \beta + A_{ZS} \gamma\\ {\bf 0}_q \end{bmatrix}.
\]
\end{example} 

The behaviour of the projection of the unfair solution onto the space of fair functions is given by the following theorem 
 \begin{theorem}
 Under Assumptions [A1] to [A3], the fair projection estimator is such that
 \begin{equation}
 \| \hat{\varphi}_{\alpha,F} - P \varphi_\dagger \|^2=O\left(\frac{1}{\alpha \delta_n} + \alpha^\beta \right)
 \end{equation}
 \end{theorem}
 \begin{proof}
 \begin{align*}
 \| \hat{\varphi}_{\alpha,F} - P \varphi_\dagger \| & \leq \| P \hat{\varphi}_\alpha - P \varphi_\dagger \| \\
 & \leq \| \hat{\varphi}_\alpha - \varphi_\dagger \| 
 \end{align*} since $P$ is a projection. The term $\| \hat{\varphi}_\alpha - \varphi_\dagger \|$ is the usual estimation term for the structural IV inverse problem. As proved in \cite{darolles2011} this term converges with the following rate of convergence 
 \[ \| \hat{\varphi}_\alpha - \varphi_\dagger \|^2 =O \left(\frac{1}{\alpha \delta_n} + \alpha^\beta \right), \] which proves the result.
 \end{proof}
 The estimator converges towards the fair part of the function $\varphi_\dagger$, i.e its projection onto the Kernel of the fairness operator $F$. If we consider the difference with respect to the usual solution we have that 
 $$ \| \hat{\varphi}_\alpha - \varphi_\dagger \|^2 =O \left(\frac{1}{\alpha \delta_n} + \alpha^\beta + \| \varphi_\dagger - P \varphi_\dagger\|^2 \right).$$
 Hence the difference $ \| \varphi_\dagger - P \varphi_\dagger\|^2 $ corresponds to the price to pay for ensuring fairness of the solution, which is null only if the true function satisfies the fairness constraint. This difference between the underlying function $\varphi_\dagger$ and its fair representation is the necessary change of the model that would enable a fair decision process minimizing the quadratic distance between the fair and the unfair functions.
 
 \subsection{Fair solution of the structural IV equation}
 A second and alternative solution to impose fairness is to solve directly the structural IV equation on the fairness space $\mathcal{N}(F)$. We denote by $K_F$ the operator $K$ restricted to $\mathcal{N}(F)$, 
 $ K_F : \mathcal{N}(F) \mapsto \mathcal{F}$. Since $\mathcal{N}(F)$ is a convex closed space, the projection onto this space is well defined and unique. We will write $P$ the projection onto $ \mathcal{N}(F) $ and $P^\perp$ the projection onto its orthogonal complement in $\mathcal{E}$, $ \mathcal{N}(F)^\perp$. \\ With these notations, we get that $K_F= K P$. \\
 \begin{definition}
Define $\varphi_{K_F}$ as the solution of the the structural equation $K \varphi = r$ in the set of fair functions defined as the kernel of the operator $F$, i.e 
 \[ {\varphi}_{K_F} = {\rm arg}\min_{\varphi \in \mathcal{N}(F)} \left( \| r-K\varphi\|^2 \right). \]
\end{definition}
Note that $\varphi_{K_F}$ is the projection of $\varphi_\dagger$ onto $\mathcal{N}(F)$ with the metric defined by $K^{*}K$, since 
 \[ {\varphi}_{K_F} = {\rm arg}\min_{\varphi \in \mathcal{N}(F)} \left( \| K \varphi_\dagger-K\varphi\|^2 \right). \]
 Note that this approximation  depends not only on $K$ but on the properties of the fair kernel $K_F=KP$. So the fairness is here quantified using its effect through the operator $K$ and thus we have called it $\varphi_{K_F}$ to highlight this dependency since the solution depends on $K$ and on $F$.\vskip .1in
 The following proposition proposes an explicit expression of ${\varphi}_{K_F} $.
 \begin{proposition}
 $${\varphi}_{K_F} = ( K_F^*K_F)^{-1} K_F^* {r}.$$
 \end{proposition}
 \begin{proof} First note that $ {\varphi}_{K_F} $ belongs to $ \mathcal{N}(F) $. For any function $g\in \mathcal{E}$, $PK^*K g \in \mathcal{N}(F) $ so the operator $( K_F^*K_F)^{-1}=(PK^{*}KP)^{-1}$ is defined from $\mathcal{N}(F) \mapsto \mathcal{N}(F) $. \\
 Let $\psi \in \mathcal{N}(F) $ so $P \psi = \psi$. We have that 
 \begin{align*}
0& = <r - K \varphi_{K_F}, K \psi> \\
 & = <K^{*} r - K^{*} K \varphi_{K_F} ,  \psi> \\
 & = <K^{*} r - K^{*} K P \varphi_{K_F} , P \psi> \\
 & = <P K^{*} r - P K^{*} K P \varphi_{K_F}, \psi>
 \end{align*} which holds for $P K^{*} r - P K^{*} K P \varphi_{K_F}=0$ which leads to ${\varphi}_{K_F} = ( PK^*K P)^{-1} P K^* {r}.$
 \end{proof}
 
 \begin{example}[Linear model, continued.]
For both our definitions of fairness in \ref{def:fairness1} and \ref{def:fairness2}, we have that $$\varphi_{K_F} = \left( P\Sigma^\prime_{XW} \Sigma_{XW} P \right)^{-1} P\Sigma^\prime_{XW} E\left[ W Y\right],$$ which simply restricts the conditional expectation operators onto the null space of $F$. 

In the case of definition \ref{def:fairness2}, the closed form expression of this estimator is easy to obtain and it is equal to
\[
\varphi_{K_F} = \begin{pmatrix} \left( \Sigma^\prime_{ZW} \Sigma_{ZW} \right)^{-1} \Sigma^\prime_{ZW} E \left[ W Y \right] \\ {\bf 0}_q \end{pmatrix} = \left( P\Sigma^\prime_{XW} \Sigma_{XW} P \right)^{-1} P\Sigma^\prime_{XW} E\left[ W Y\right],
\]
which is equivalent to exclude $S$ from the second stage estimation of the IV model, and where
\[
F = \begin{bmatrix} {\bf 0}_{p \times p} & {\bf 0}_{p \times q} \\ {\bf 0}_{q \times p} & I_q \end{bmatrix}, \text{ and } P = I_{p+q} - F. 
\]
 \end{example}
 
Now consider the fair approximation of the solution of \eqref{eq:inv}  as the solution of the following minimization program \[ \hat{\varphi}_{K_F,\alpha} = {\rm arg}\min_{\varphi \in \mathcal{N}(F)} \left( \| \hat{r}-K\varphi\|^2 + \alpha \| \varphi\|^2 \right). \]
\begin{proposition}
The fair solution of the IV structural equation has the following expression
\[ \hat{\varphi}_{K_F,\alpha} = (\alpha {\rm Id}+ K_F^*K_F)^{-1} K_F^* \hat{r}. \]
It converges to $\varphi_{K_F}$ when $\alpha$ goes to zero as soon as $\alpha$ is chosen such that $\alpha \delta_n \rightarrow + \infty$.
\end{proposition} 
\begin{proof}
As previously, $\hat{\varphi}_{K_F,\alpha} $ minimizes in $\mathcal{N}(F)$, $ \| \hat{r}-K\varphi\|^2 + \alpha \| \varphi\|^2$. Hence the first order condition is that for all $g \in \mathcal{N}(F)$ we have 
\begin{align*} <-Kg,\hat{r}-K \varphi> + \alpha <g,\varphi> & = 0 \\
<g,K^{*}K\varphi-K^* \hat{r}>+ \alpha<g,\varphi> & = 0 \\
<g, P K^{*}K\varphi-P K^* \hat{r} + \alpha \varphi >& =0.
\end{align*}
Hence using $K^*_F= PK^*$ and since $\varphi$ is in $\mathcal{N}(F)$ and thus $P \varphi = \varphi$, we obtain the expression of the theorem. \\
Using this expression we can compute the estimation as follows :
\begin{align*} & \hat{\varphi}_{K_F,\alpha} -\varphi_{K_F}= \\
& (\alpha {\rm Id}+ K_F^*K_F)^{-1} K_F^* (\hat{r}-K \varphi_\dagger) + ((\alpha {\rm Id}+ K_F^*K_F)^{-1} - (K_F^*K_F)^{-1} ) K_F^* K \varphi_\dagger \\
& = (I) \quad + \quad (II). \end{align*}
The first term is a variance term which is such that 
$$ \| (I) \|^2 = O \left( \frac{1}{\alpha \delta_n} \right).$$
 Recall that for two operators $$ A^{-1}-B^{-1}= A^{-1}(B-A)B^{-1} $$ 
 Hence the second term can be written as
 $$(II)= -\alpha (\alpha {\rm Id}+ K_F^*K_F)^{-1} \varphi_{K_F} .$$ This tern is the bias of Tikhonov's regularization of the operator $K_F^*K_F = PK^*K P$
 which goes to zero when $\alpha$ goes to zero.
 \end{proof}
 When $\alpha$ decreases to zero, the rate of consistency of the projected fair estimator can be made precise if we assume some Hilbert scale regularity for both the fair part of $\varphi_\dagger$ and the remaining {\it unfair} part $P^\perp \varphi_\dagger$. \vskip .1in
 Assume that 
 
\begin{itemize}
\item {\bf [E1]} $ P \varphi_\dagger \in \mathcal{R}(PK^*KP)^\frac{\beta}{2}$ for $\beta \leq 2$
\item {\bf [E2]} $ P^\perp \varphi_\dagger \in \mathcal{R}(PK^*KP)^\frac{\gamma}{2}$ for $\gamma \leq 2$.
\end{itemize} 
Previous assumptions are analogous to the source condition {\bf [A3]} adapted to the fair operator $K_F$.
 \begin{theorem}
Under Assumptions {\bf [E1]} and {\bf [E2]}, the estimator $\hat{\varphi}_{K_F}$ converges towards $\varphi_{K_F}$ at the following rate
\[ \| \hat{\varphi}_{K_F}- \varphi_{K_F} \|^2 = O \left( \frac{1}{\alpha \delta_n} +\alpha^{min(\beta,\gamma)}\right) \]
 \end{theorem}
 We recognise the usual rate of convergence of the Tikhonov's regularized estimator. The main change is given here by the fact that the rate is driven by the fair source conditions {\bf [E1]} and {\bf [E2]} which relates the smoothness of the function with the decay of the SVD of the kernel restricted to the the kernel of the fairness operator.
\begin{proof}
The rate of consistency depends on the term $(II)$ defined previously. We decompose here into two terms.
\begin{align*} (II) & = -\alpha (\alpha {\rm Id}+ K_F^*K_F)^{-1}(K_F^*K_F)^{-1} K_F^* (K P \varphi_\dagger + KP^\perp \varphi_\dagger) \\
& = (A) \quad + \quad (B).
\end{align*}
First remark that since $P=P^2$
\begin{align*} (A) & =-\alpha (\alpha {\rm Id}+ K_F^*K_F)^{-1}(K_F^*K_F)^{-1} K_F^*K_F P \varphi_\dagger \\
& = -\alpha (\alpha {\rm Id}+ K_F^*K_F)^{-1} P \varphi_\dagger \end{align*}
Assumption {\bf [E1]} provides the rate of decay of this term $\| (A) \|^2$ and enables to prove that it is of order $\alpha^\beta$.\\
For the second term $(B)$, consider the SVD of the operator $K_F=KP$ denoted by $\lambda_j,\psi_j,e_j$ for all $j \geq 1$. So we have that 
\begin{align*} \| (B) \|^2 & = \|  \alpha (\alpha {\rm Id}+ K_F^*K_F)^{-1}(K_F^*K_F)^{-1} K_F^*K P^\perp \varphi_\dagger \|^2 \\
& = \alpha^2 \sum_{j \geq 1} \frac{\lambda_j^2}{\lambda_j^4 (\alpha+ \lambda_j^2)^2} |<K P^\perp \varphi_\dagger ,e_j>
|^2 \\
& = \alpha^2 \sum_{j \geq 1} \frac{\lambda_j^{2\gamma}}{ (\alpha+ \lambda_j^2)^2} \frac{ | <K P^\perp \varphi_\dagger ,e_j>
|^2}{\lambda_j^{2(1+\gamma)} }	 \\
& = O(\alpha^\gamma )
\end{align*} 
To ensure that $$ \sum_{j \geq 1} \frac{ | <K P^\perp \varphi_\dagger ,e_j>|^2}{\lambda_j^{2(1+\gamma)}} < + \infty$$ we assume that 
$$ \sum_{j \geq 1} \frac{ | < P^\perp \varphi_\dagger , \lambda_j \psi_j>|^2}{\lambda_j^{2(1+\gamma)}} = \sum_{j \geq 1} \frac{ | < P^\perp \varphi_\dagger , \psi_j>|^2}{\lambda_j^{2\gamma}} < + \infty$$ where $K^* e_j= \lambda_j \psi_j$, which is ensured under Assumption {\bf [E2]}. Finally the two terms are of order $O(\alpha^\beta + \alpha^\gamma)$, which proves the result.
\end{proof}

In conclusion we have defined two fair approximations of the function $\varphi_\dagger$. The first one is its fair projection $\varphi_F=P\varphi_\dagger$ while the other is the solution of the fair kernel $\varphi_{K_F}$. The two solutions coincide as soon as $$ \varphi_{K_F}-P\varphi_\dagger=(K_F^*K_F)^{-1}K_F^*KP^{\perp} \varphi_\dagger =0. $$ Under assumption [A2], $K^*_FK_F$ is also one to one, hence the difference between both approximations is null only if \begin{equation} \label{eq:lesmemes} K P^\perp \varphi_\dagger =0. \end{equation}
If we consider the case of (IV) regression. This condition is met as soon as 
$$ E ( \varphi(Z,S)|W) - E ( E ( \varphi(Z,S)|Z) | W) =0.$$ This is the case when the sensitive variable $S$ is independent w.r.t to the instrument $W$ conditionally to the characteristics $Z$. Yet in the general case, both functions are different.

\subsection{Approximate fairness}
 Imposing \eqref{eq:fullfair} is a way to ensure complete fairness of the solution of \eqref{eq:inv}. In many cases, this complete fairness leads to bad approximation properties, hence it is replaced by a constraint on the norm of $F \varphi$. Namely we look for the estimator defined as the solution of the optimization 
 \begin{equation} \label{eq:fairoptimnoise}
\hat{\varphi}_{\alpha,\rho} = {\rm arg}\min_{\varphi \in \mathcal{E}} \left( \|\hat{r}-K \varphi\|^2 + \alpha \| \varphi \|^2 +\rho \| F \varphi\|^2 \right)
 \end{equation}
This estimator corresponds to the usual Tikhonov regularized estimator with an extra penalty term $\rho \| F \varphi\|^2 $. The penalty enforces fairness since it enforces $\| F\varphi \|$ to be small which corresponds to a relaxation of the full fairness constraint $F\varphi =0$. The parameter $\rho$ provides a trade-off between the level of fairness which is imposed and the closeness to the usual estimator of the non parametric IV function. \\
We study its asymptotic behaviour in the following theorem. \\ Note first that the solution of \eqref{eq:fairoptimnoise} has a close form and can be written as 
 \[ \hat{\varphi}_{\alpha,\rho}=(\alpha {\rm Id}+ \rho F^*F + K^*K)^{-1}K^* \hat{r}. \] \\
 
The asymptotic behaviour of the estimator is provided by the following theorem. It also ensures that the limit solution of \eqref{eq:fairoptimnoise}, i.e when $\rho \rightarrow + \infty$, is fair in the sense that $\lim_{\rho \rightarrow + \infty} \|F\varphi_{\alpha,\rho} \| =0.$ It converges to the solution of the structural solution restricted to the set of fair functions $\varphi_{K_F}$. \\
We will use the following notations. Consider the collection of operators $$ L_\alpha=(\alpha {\rm Id} + K^*K)^{-1} F^*F $$ 
$$ L=( K^*K)^{-1} F^*F. $$
\begin{itemize}
\item{\bf [A4]} $\mathcal{R}(F^*F) \subset \mathcal{R}(K^*K)$. This condition guarantees that the operators $ L$ and $L_\alpha$ are well defined operators.
\end{itemize}
$L$ is an operator $T: \mathcal{E} \rightarrow \mathcal{E}$ which is not self-adjoint. \\
Consider also the operator 
$$ T = (K^*K)^{-1/2}F^*F (K^*K)^{-1/2} $$ which is an self-adjoint operator which is well defined as soon as 
\begin{itemize}
\item{\bf [A5]} $\mathcal{R}(F^*F) \subset \mathcal{R}(K^*K)^{1/2}$. 
\end{itemize}
If we assume a source condition on the form \begin{itemize} 
\item{\bf [A6]} There exists $\gamma \geq \beta$ $$ F^*F P^\perp \varphi_\dagger \in \mathcal{R} (K^*K)^\frac{\gamma +1}{2}$$ \end{itemize}

 \begin{theorem}[Consistency of fair IV estimator] \label{thm:main} 
 The approximated fair IV estimator $\hat{\varphi}_{\alpha,\rho}$ is an estimator of the fair projection of the structural function, i.e $\varphi_{K_F}$. Its rate of consistency under assumptions [A1] to [A6] is given by
 
 \begin{equation} \label{eq:mainrate} \| \hat{\varphi}_{\alpha,\rho} - \varphi_{K_F} \|^2= O \left( \alpha^\beta + \frac{1}{\rho^2} + \frac{1}{\alpha \delta_n} \right). 
 \end{equation}
 \end{theorem}
The rate of convergence is consistent in the following sense. When we increase the level of imposed fairness to the full fairness constraint, i.e when $\rho$ goes to infinity, for appropriate choices of smoothing parameter $\alpha$, the estimator converges to a full fair function. The rate in $\frac{1}{\rho^2}$ corresponds to the fairness part of the rate. If $\beta$ the Source condition parameter can be chosen large enough such that $\alpha^\beta = \frac{1}{\rho^2}$, hence we recover, for an optimal choice of $\alpha_{\rm opt}$ of order $\delta_n^{-\frac{1}{\beta+1}}$, the usual rate of consistence of non parametric IV estimates $$ \| \hat{\varphi}_{\alpha,\rho} - \varphi_{K_F} \|^2=O \left(\delta_n^{-\frac{\beta}{\beta+1}} \right).$$

 \begin{example}[Linear model, continued.]
In the linear IV model, let $$\varphi_\rho = \left( \rho F^\prime F + \Sigma^\prime_{XW} \Sigma_{XW} \right)^{-1} \Sigma^\prime_{XW} E\left[ W Y\right],$$ the estimator which imposes the approximate fairness constraint. Notice that
\begin{align*}
& \left( \rho F^\prime F + \Sigma^\prime_{XW} \Sigma_{XW} \right)^{-1} \\
\qquad &= \left(\Sigma^\prime_{XW} \Sigma_{XW} \right)^{-1} - \rho \left(\Sigma^\prime_{XW} \Sigma_{XW} \right)^{-1} F^\prime \left( I_q + \rho F \left(\Sigma^\prime_{XW} \Sigma_{XW} \right)^{-1} F^\prime \right)^{-1} F \left(\Sigma^\prime_{XW} \Sigma_{XW} \right)^{-1}\\
\qquad &= \left(\Sigma^\prime_{XW} \Sigma_{XW} \right)^{-1} - \left(\Sigma^\prime_{XW} \Sigma_{XW} \right)^{-1} F^\prime \left( \frac{1}{\rho}I_q + F \left(\Sigma^\prime_{XW} \Sigma_{XW} \right)^{-1} F^\prime \right)^{-1} F \left(\Sigma^\prime_{XW} \Sigma_{XW} \right)^{-1}.
\end{align*}

This decomposition implies that 
\[
\lim_{\rho \rightarrow \infty} \varphi_\rho = \varphi - \left(\Sigma^\prime_{XW} \Sigma_{XW} \right)^{-1} F^\prime \left( F \left(\Sigma^\prime_{XW} \Sigma_{XW} \right)^{-1} F^\prime \right)^{-1} F \varphi,
\]
which directly gives $$\lim_{\rho \rightarrow \infty} F \varphi_\rho = 0.$$ Therefore, as implied by our general theorem, as $\rho$ diverges to $\infty$, the full fairness constraint is imposed.
 \end{example}
 
\begin{remark}
Previous theorems enable to understand the asymptotic behaviour of the fair regularized IV estimator. When $\alpha$ goes to zero but $\rho$ is fixed, it converges towards towards a function $\varphi_\rho$ which differs from the original function $\varphi_\dagger$ that could have been estimated without fairness constraint. Interestingly we point out that the constraint on fairness enables to obtain a fair solution but that the solution is not the fair approximation of the original function $\varphi_\dagger$. Rather the fair solution is obtained by considering the set of approximated solutions which satisfy to the fairness constraint.

 \end{remark}
 
\begin{remark}

The theorem requires an additional assumption denoted by {\bf [A6]}. This assumption aims at controlling the regularity of the {\it unfair} part of the function $\varphi_\dagger$. It is analogous to a source condition imposed on the part of the solution which does not lie in the kernel of the operator which models the set of fair functions, namely $P^\perp \varphi_\dagger$. This condition is obviously fulfilled if $\varphi_\dagger$ is fair since $P^\perp \varphi_\dagger=0$.
\end{remark}

\begin{remark}
The smoothness assumptions we impose in this paper are  source conditions with regularity smaller than 2. Such restrictions come from the choice of  standard Tikhonov's  regularization method. Choosing other methods such as Landwebers's iteration or iterated Tikhonov's regularization would enable to deal with more regular functions, without changing the results presented in this work.
\end{remark}

 Proof of Theorem~\eqref{thm:main}
 \begin{proof}
 Note that the fair IV estimator can be decomposed into a bias and a variance term that will be studied separately
 \begin{align*}
 \hat{\varphi}_{\alpha,\rho} & =(\alpha {\rm Id}+ \rho F^*F + K^*K)^{-1}K^* \hat{r} \\
 & = (\alpha {\rm Id}+ \rho F^*F + K^*K)^{-1}K^*r + (\alpha {\rm Id}+ \rho F^*F + K^*K)^{-1}K^*U_n \\
 & = (B) +( V) .
 \end{align*}
 Then the bias term can be decomposed as
 \begin{align*}
(B) & = [(\alpha {\rm Id}+ \rho F^*F + K^*K)^{-1} - (\rho F^*F + K^*K)^{-1}] K^*r + (\rho F^*F + K^*K)^{-1} K^*r \\
& = (B_1) + (B_2). 
 \end{align*}
 The operator $(\alpha {\rm Id}+ \rho F^*F + K^*K)^{-1}$ can be written as 
 \begin{align*}
 (\alpha {\rm Id}+ \rho F^*F + K^*K)^{-1} & = (R^{-1}_\alpha(K) + \rho F^*F)^{-1} \\
 & = ({\rm Id} + \rho R_\alpha(K)F^*F)^{-1} R_\alpha(K)
 \end{align*}
 Note that condition [A4] ensures that $$L_\alpha := R_\alpha(K)F^*F = (K^*K+\alpha {\rm Id})^{-1} F^*F $$ is a well defined operator on $\mathcal{E}$. Moreover condition [A2] ensures that $R_\alpha(K)$ is one to one hence the kernel of the operator $L_\alpha$ is the kernel of $F$. Hence we have using the Tikhonov approximation \eqref{tiko}
 \begin{align*}
 (B1) & = [({\rm Id} + \rho L_\alpha)^{-1} R_\alpha(K) - ({\rm Id} + \rho L )^{-1} (K^*K)^{-1}] K^* r \\
 & =({\rm Id} + \rho L_\alpha)^{-1} (\varphi_\alpha - \varphi_\dagger) + [({\rm Id} + \rho L_\alpha)^{-1}- ({\rm Id} + \rho L )^{-1}] \varphi_\dagger \end{align*}
We will study each term separately. \vskip .1in
\noindent
$\bullet$ Since $\|({\rm Id} + \rho L_\alpha )^{-1}\|$ is bounded we get that the first term is of order the rate of convergence of $\varphi_\alpha -\varphi_\dagger$. Hence under source condition [A3] we have that
$$ \| ({\rm Id} + \rho L_\alpha )^{-1} (\varphi_\alpha -\varphi_\dagger)\|^2 = O(\alpha^\beta).$$ \vskip .1in
 
\vskip .1in
\noindent $\bullet$ Using that for two operators $$ A^{-1}-B^{-1}= A^{-1}(B-A)B^{-1} $$ we obtain for the second term that
$$ \left( ({\rm Id} + \rho L_\alpha)^{-1} - ({\rm Id} + \rho L )^{-1} \right) \varphi_\dagger = \rho ({\rm Id} + \rho L_\alpha)^{-1}(L-L_\alpha) ({\rm Id} + \rho L )^{-1} \varphi_\dagger. $$
Note that $(L-L_\alpha) P\varphi_\dagger =0 $ and $({\rm Id} + \rho L )^{-1} P \varphi_\dagger = P \varphi_\dagger $ hence we can replace $\varphi_\dagger$ in the last expression by the projection onto the orthogonal space to the kernel, namely $P^\perp \varphi_\dagger .$ Hence 
$$ \| \left( ({\rm Id} + \rho L_\alpha)^{-1} - ({\rm Id} + \rho L )^{-1} \right) \varphi_\dagger \|^2 =O \left( \rho^2 \| L-L_\alpha \|^2 \| ({\rm Id} + \rho L )^{-1} P^\perp \varphi_\dagger \|^2 \right)$$
We have that $ ({\rm Id} + \rho L )^{-1} P^\perp \varphi_\dagger \|^2= O (1/\rho^2)$. 
Then $$ L-L_\alpha = \alpha (\alpha {\rm Id} + K^*K)^{-1} ( K^*K)^{-1} F^*F .$$
Under Assumption {\bf [E6]}, We obtain that $( K^*K)^{-1} F^*F P^\perp \varphi_\dagger$ is of regularity $\gamma$ so 
$$\| (L-L_\alpha ) P^\perp \varphi_\dagger \|^2 =O \left( \alpha^\gamma \right).$$
Hence we can conclude that 
$$ \| \left( ({\rm Id} + \rho T_\alpha)^{-1} - ({\rm Id} + \rho T )^{-1} \right) \varphi_\dagger \|^2= O \left( \alpha^\gamma \right).$$
The second term $(B_2)$ is such that $(B_2)= (\rho F^*F + K^*K)^{-1} K^*r.$ We can write
\begin{align*}
(B_2) & = \left((K^*K)^{1/2} ({\rm Id} + \rho (K^*K)^{-1/2}F^*F (K^*K)^{-1/2}) (K^*K)^{1/2} \right)^{-1} K^*K\varphi_\dagger \\
& = (K^*K)^{-1/2} ({\rm Id} +\rho T)^{-1} (K^*K)^{1/2} \varphi_\dagger, 
\end{align*}
where $T:= (K^*K)^{-1/2}F^*F (K^*K)^{-1/2} $ is a self-adjoint operator well defined using Assumption {\bf [A5]}. Let $$ \varphi_\rho= (K^*K)^{-1/2} ({\rm Id} +\rho T)^{-1} (K^*K)^{1/2} \varphi_\dagger.$$
\begin{itemize}
\item Note first that $\varphi_\rho$ converges when $\rho \rightarrow + \infty$ to the the projection of $\psi:= (K^*K)^{1/2} \varphi_\dagger$ onto ${\rm Ker} (T).$ As a matter of fact we can write the SVD of $T$ as $\lambda_j^2$ and $e_j$ for $j \geq 1$. So we get that
\begin{align*}
({\rm Id} +\rho T)^{-1} \psi& = \sum_{j \geq 1} \frac{1}{1+ \rho \lambda_j^2} <\psi,e_j> \\
 & = \sum_{j \geq 1, \lambda_{j} \neq 0 } \frac{1}{1+ \rho \lambda_j^2} <\psi,e_j> e_j + \sum_{j \geq 1, \lambda_{j} = 0 } <\psi, e_{j}> e_{j} .
 \end{align*}
 The last quantity converges when $\rho \rightarrow + \infty$ towards the projection of $\psi$ onto the kernel of $T$. Applying the operator $(K^*K)^{-1/2}$ does not change the limit since $K^*K$ is one to one.
\item Note then that the kernel of the operator $T$ can be identified as follows
\begin{align*}
 \{ \psi \in {\rm Ker}(T) \} & =\{ \psi, \quad F (K^*K)^{-1/2} \psi=0 \} \\
 & =\{ \psi, \quad (K^*K)^{-1/2} \psi \in {\rm Ker}(F) \} \\
 & = \{ \psi=(K^*K)^{1/2} \varphi, \quad \varphi \in {\rm Ker}(F) \}. 
\end{align*}
Hence $\varphi_\rho$ converges towards the projection of $(K^*K)^{1/2} \varphi_\dagger$ onto the functions $(K^*K)^{1/2} \varphi$ with $ \varphi \in {\rm Ker}(F)$.
\item Characterization of the projection. Note that the projection can be written as
\begin{align*}
& {\rm arg}\min_{\varphi \in {\rm Ker}(F)} \| (K^*K)^{1/2} \varphi_\dagger - (K^*K)^{1/2} \varphi \|^2 \\
& = {\rm arg}\min_{\varphi \in {\rm Ker}(F)} \| (K^*K)^{1/2} (\varphi_\dagger - \varphi) \|^2 \\
& = {\rm arg}\min_{\varphi \in {\rm Ker}(F)} < (K^*K)^{1/2} (\varphi_\dagger - \varphi), (K^*K)^{1/2} (\varphi_\dagger - \varphi)> \\
& = {\rm arg}\min_{\varphi \in {\rm Ker}(F)} <\varphi_\dagger - \varphi, (K^*K)(\varphi_\dagger - \varphi)> \\
& = {\rm arg}\min_{\varphi \in {\rm Ker}(F)} \| K(\varphi_\dagger - \varphi) \|^2 \\
& = {\rm arg}\min_{\varphi \in {\rm Ker}(F)} \| r - K \varphi \|^2 \\
& = \varphi_{K_F}
\end{align*} as defined previously.
\item Finally usual bounds enable to prove that $$ \| \varphi_\rho- \varphi_{K_F} \|^2 = O \left(\frac{1}{\rho^2} \right).$$
\end{itemize}
Using all previous bounds, we can write 
 \begin{equation} \label{eq:inbetween} \| (B)-P\varphi_\dagger \|^2=O(\frac{1}{\rho^2}+ \alpha^\beta +\alpha^\gamma ) .\end{equation}

 Finally we prove that the variance term $(V)$ is such that 
 
 \[ \| (\alpha {\rm Id}+ \rho F^*F + K^*K)^{-1}K^*U_n \|^2 = O\left(\frac{1}{\alpha \delta_n} \right)
 \]
 Actually, using previous notations, we get that 
 \begin{align*}
 \| (\alpha {\rm Id}+ \rho F^*F + K^*K)^{-1}K^*U_n \|& = \|({\rm Id} + \rho L_\alpha)^{-1} (\alpha {\rm Id} + K^*K)^{-1} K^* U_n\| \\
 & \leq \| ({\rm Id} + \rho L_\alpha)^{-1} \| \| (\alpha {\rm Id} + K^*K)^{-1} K^* \| \|U_n \| \\
 & \leq \| ({\rm Id} + \rho L_\alpha)^{-1} \| \frac{1}{\alpha} \frac{1}{\delta_n^{1/2}}.
 \end{align*}
 Using that $ ({\rm Id} + \rho L_\alpha)^{-1} $ is bounded leads to the desired result. \\
Both bounds prove the final result for the theorem.
 \end{proof}
 Choosing the fairness constraint implies to modify the usual IV estimator. The following theorem quantifies at fixed $\rho$ and $\alpha$ the deviation of the fair IV estimator \eqref{eq:fairoptimnoise} with respect to the unfair natural solution of the IV problem. 
 \begin{theorem}[Price for fairness]
 \[ \| \varphi_\alpha- \varphi_{\alpha,\rho} \| =O \left( \frac{\rho}{\alpha^2} \right)\]
 \end{theorem}
 
 \begin{proof}
 \begin{align*}
& \| \varphi_\alpha- \varphi_{\alpha,\rho} \| \\
& \leq \| (\alpha {\rm Id} + K^*K)^{-1}K^* r - (\alpha {\rm Id}+ \rho F^*F + K^*K)^{-1}K^* r \| \\
& \leq \| [(\alpha {\rm Id} + K^*K)^{-1} - (\alpha {\rm Id}+ \rho F^*F + K^*K)^{-1}] K^*r \|.
\end{align*}
Using that for two operators $$ A^{-1}-B^{-1}= A^{-1}(B-A)B^{-1} $$ we obtain
 \begin{equation}
 \| \varphi_\alpha- \varphi_{\alpha,\rho} \| \leq \| (\alpha {\rm Id}+ \rho F^*F + K^*K)^{-1} \rho F^*F (\alpha {\rm Id} + K^*K)^{-1} \|
 \end{equation}
 Now using that \begin{align*}
 \|(\alpha {\rm Id} + K^*K)^{-1} \| & \leq \frac{1}{\alpha} \\
 (\alpha {\rm Id}+ \rho F^*F + K^*K)^{-1} & \leq \frac{1}{\alpha} 
 \end{align*} and since $$ \| K^* r \| \leq M $$ leads to the result.
 \end{proof}
 
 Previous theorem suggests that in a decision procedure, the stakeholder should choose make a choice : imposing fairness conditions and obtaining an approximately  {\it fair} solutions provides a different solution than the usual estimates, more different as $\rho$, the weight put on the fairness penalty, increases. This cost for changing previous uses for a new fair model, could be included in the decision process as soon as we could define an  the economic value for fairness. In this framework, this would provide a balance between similarity with the {\it unfair} usual model and the desired level of fairness that  could be used to optimize the choice of the trade-off parameter $\rho$.
 
\section{Estimation with an exogenous binary sensitive attribute} \label{sec:estimation}

We discuss the estimation and the finite sample implementation of our method in the simple case when $S$ is an exogenous binary random variable (for instance, gender or race), and $Z \in \IR^p$ only contains continuous endogenous regressors. This framework can be easily extended to the case when $S$ is an endogenous multivariate categorical variable and to include additional exogenous components in $Z$ \cite{hall2005,centorrino2013a,centorrino2017}. 
Our statistical model can be written as 

\begin{equation} \label{eq:statmod1}
Y = \varphi_0(Z) + \varphi_1(Z) S + U = {\bf S}^\prime \varphi(Z) + U, 
\end{equation}
where $\varphi = [\varphi_0 \quad \varphi_1]^\prime$, and $\mathbf{S} = [1 \quad S]^\prime$. 

This model is a varying coefficient model see, among others,\cite{HASTIE_TIBSHIRANI:1993,FAN_ZHANG:1999,LI_HUANG_LI_FU:2002}. Adopting the terminology that is used in this literature, we refer to ${\bf S}$ as the `linear' variables (or predictors), and to the $Z$'s as the `smoothing' variables (or covariates) \cite{FAN_ZHANG:2008}. When $Z$ is endogenous, \cite{centorrino2017} have studied identification and estimation of this model with instrumental variables. That is, we assume there is a random vector $W \in \IR^q$, such that $E\left[ {\bf S} U \vert W\right] = 0$, and
\begin{equation}\label{eq:completevc}
E\left[ {\bf S}{\bf S}^\prime \varphi(Z) \vert W \right] = 0 \quad \Rightarrow \quad \varphi = 0,
\end{equation}
where equalities are intended almost surely. {Notice that the moment conditions $E\left[ {\bf S} U \vert W\right] = 0$ are implied by the assumption that $E\left[ U \vert W,S\right] = 0$, although they allow to exploit the semiparametric structure of the model, and reduce the curse of dimensionality \cite{centorrino2017}. The \textit{completeness} condition in equation \eqref{eq:completevc} is a necessary condition for identification, and it is assumed to hold. As proven in \cite{centorrino2017}, this condition is implied by the injectivity of the conditional expectation operator (see our Assumption A2), and by the matrix $E[{\bf S}{\bf S}^\prime \vert z,w]$ being full rank for almost every $(z,w)$.\

We would like to obtain a nonparametric estimator of the functions $\lbrace \varphi_0,\varphi_1 \rbrace$ when a fairness constraint is imposed. We use the following operator's notations
\begin{align*}
\left( K_s\varphi \right) (w)=& E\left[ \mathbf{S}\mathbf{S}^\prime \varphi(Z) \vert W = w \right] \\
\left( K_s^\ast \psi \right) (z)=& E\left[ \mathbf{S}\mathbf{S}^\prime \psi(W) \vert Z = z\right]\\
\left( K^\ast \psi \right) (z)=& E\left[ \psi(W) \vert Z = z\right],
\end{align*}
for every $\varphi \in L^2(Z)$, and $\psi \in L^2(W)$.

When no fairness constraint is imposed, the regularized approximation to the pair $\lbrace \varphi_0,\varphi_1 \rbrace$ is given by 
\begin{equation} \label{eq:minprob}
\varphi_{\alpha} = {{\rm arg}\min}_{\varphi \in L^2(Z)} \Vert K_s \varphi - r \Vert^2 + \alpha \Vert \varphi \Vert^2,
\end{equation}
where $\Vert \varphi \Vert^2 = \Vert \varphi_0 \Vert^2 + \Vert \varphi_1 \Vert^2$. That is
\begin{equation}\label{eq:phialphavc}
\varphi_{\alpha} = \left( \alpha I + K_s^\ast K_s \right)^{-1} K_s^\ast r, 
\end{equation}
with $r(w) = E\left[ {\bf S} Y \vert W = w\right]$.

As in \cite{centorrino2017}, the quantities in equation \eqref{eq:phialphavc} can be replaced by consistent estimators. Let $\lbrace (Y_i,X_i,W_i), i = 1,\dots,n \rbrace$ an iid sample from the joint distribution of $(Y,X,W)$. We denote by 
\[
{\bf Y}_n =\begin{bmatrix} Y_1 \\ Y_2 \\ \vdots \\ Y_n \end{bmatrix}, \qquad {\bf S}_n =\begin{bmatrix} I_n & diag(S_1,S_2,\dots,S_n) \end{bmatrix},
\]
the $n \times 1$ vector which stacks the observations of the dependent variable and the $n \times 2n$ matrix of predictors, where $I_n$ is the identity matrix of dimension $n$, and $diag(S_1,S_2,\dots,S_n)$ is a $n\times n$ diagonal matrix, whose diagonal elements are equal to the sample observations of the sensitive attribute $S$. Similarly, we let
\[
{\bf D}_{1,n} = \begin{bmatrix} S_1 \\ S_2 \\ \vdots \\ S_n \end{bmatrix}, \text{ and } {\bf D}_{0,n} = \begin{bmatrix} 1 - S_1 \\ 1 - S_2 \\ \vdots \\ 1- S_n \end{bmatrix},
\]
two $n \times 1$ vectors stacking the sample observations of $S$ and $1 - S$.

Finally, let $C(\cdot)$ a univariate kernel function, such that $C(\cdot) \geq 0$, and $\int C(u) du = 1$, and ${\bf C}(\cdot)$ be a multivariate product kernel. That is, for a vector ${\bf u}=\begin{bmatrix} u_1 & u2 & \dots & u_p \end{bmatrix}^\prime$, with $p \geq 1$, ${\bf C}({\bf u}) = C(u_1) \times C(u_2) \times \dots \times C(u_p)$. 

As detailed in \cite{centorrino2013a}, the operators $K$ and $K^\ast$ can be approximated by finite dimensional matrices of kernel weights. In particular, we have that 
\begin{equation*}
\underbrace{\hat{K}}_{n \times n} = \begin{bmatrix} {\bf C} \left( \frac{W_i - W_j}{h_W}\right) \end{bmatrix}^n_{i,j= 1} \text{ and } \underbrace{\widehat{K^\ast}}_{n \times n} = \begin{bmatrix} {\bf C} \left( \frac{Z_i - Z_j}{h_Z}\right) \end{bmatrix}^n_{i,j= 1},
\end{equation*}
where $h_W$ and $h_X$ are bandwidth parameters. Therefore,
\begin{align*}
\hat{r} =& vec \left( (I_2 \otimes \hat{K}) {\bf S}^\prime_n {\bf Y}_n\right)\\
\hat{K}_s =& (I_2 \otimes \hat{K}) {\bf S}_n^\prime {\bf S}_n\\
\widehat{K^\ast}_s =& (I_2 \otimes \widehat{K^\ast}) {\bf S}_n^\prime {\bf S}_n,
\end{align*}
in a way that 
\begin{equation}\label{eq:vhatunc}
\hat\varphi_{\alpha} = \begin{bmatrix} \hat\varphi_{0,\alpha} & \hat\varphi_{1,\alpha} \end{bmatrix} = \left( vec(I_n)^\prime \otimes I_n \right)\left( I_n \otimes \left( \alpha I + \widehat{K^\ast}_s \hat{K}_s \right)^{-1} \widehat{K^\ast}_s \hat{r} \right). 
\end{equation}

As explained above, the fairness constrain can be characterized by a linear operator $F_j$, such that $F_j \varphi = 0$, where $j = \lbrace 1,2 \rbrace$. In case of definition \ref{def:fairness1}, and exploiting the binary nature of $S$, the operator $F_1$ can be approximated by
\[
\underbrace{F_{1,n}}_{2n \times 2n} = \begin{bmatrix} {\bf 0}_n & {\bf 0}_n \\ \iota_n \left[ \left({\bf D}^\prime_{1,n} {\bf D}_{1,n} \right)^{-1} {\bf D}^\prime_{1,n} - \left({\bf D}^\prime_{0,n} {\bf D}_{0,n} \right)^{-1} {\bf D}^\prime_{0,n} \right] & \iota_n \left({\bf D}^\prime_{1,n} {\bf D}_{1,n} \right)^{-1} {\bf D}^\prime_{1,n} \end{bmatrix},
\]
where $\iota_n$ is a $n \times 1$ vector of ones, and ${\bf 0}_n$ is a $n \times n$ matrix of zeroes. 

In the case of definition \ref{def:fairness2}, the fairness operator can be approximated by 
\[
\underbrace{F_{2,n}}_{2n \times 2n} = \begin{bmatrix} {\bf 0}_n & {\bf 0}_n \\ {\bf 0}_n & I_n \end{bmatrix},
\]
In both cases, when the function $\varphi \in \mathcal{F}_j$, we obviously have that $F_j vec(\varphi) = 0$, with $j = \lbrace 1,2 \rbrace$.

As detailed in Section \ref{sec:fullfair}, and for $j = \lbrace 1,2 \rbrace$, the estimator consistent with the fairness constraint can be obtained in several ways
\begin{enumerate}
\item[1)] By projecting the unconstrained estimator in \eqref{eq:vhatunc} onto the null space of $F_j$. Let $P_{j,n}$ be the estimator of such projection, then we have that
\begin{equation} \label{eq:vhatcon1}
\hat\varphi_{\alpha,F,j} = \left( vec(I_n)^\prime \otimes I_n \right)\left( I_n \otimes P_{j,n} vec(\hat\varphi_{\alpha}) \right),
\end{equation}
\item[2)] By restricting the conditional expectation operator to project onto the null space of $F_j$. Let
\begin{align*}
\hat{K}_{F,j,s} = \hat{K}_{s} P_{j.n}, \text{ and } \widehat{K^\ast}_{F,j,s} = P_{j,n} \widehat{K^\ast}_{s},
\end{align*}
then 
\begin{equation} \label{eq:vhatcon2}
\hat\varphi_{\alpha,K_F,j} = \left( vec(I_n)^\prime \otimes I_n \right)\left( I_n \otimes \left( \alpha I + \widehat{K^\ast}_{F,j,s} \hat{K}_{F,j,s} \right)^{-1} \widehat{K^\ast}_{F,j,s} \hat{r} \right),
\end{equation}
\item[3)] By modifying the objecting function to include an additional term which penalizes deviations from fairness. That is, we let 
\[
\hat\varphi_{\alpha,\rho,j}={{\rm arg}\min}_{\varphi \in \mathcal{F}_j} \Vert \hat{K}_s \varphi - \hat{r} \Vert^2 + \alpha \Vert \varphi \Vert^2 + \rho \Vert F_{j,n}\varphi \Vert^2,
\]
in a way that
\begin{equation}\label{eq:vhatcon3}
\hat\varphi_{\alpha,\rho,j}= \left( \alpha I_n + \rho F_{j,n} F_{j,n} + \widehat{K^\ast}_s \hat{K}_s \right)^{-1} \widehat{K^\ast}_s \hat{r}.
\end{equation}
For $\rho = 0$, this estimator is equivalent to the unconstrained estimator, $\hat\varphi_{\alpha}$, and, for $\rho$ sufficiently large it imposes the full fairness constraint.
\end{enumerate}

To implement the estimators above, we need to select several smoothing, $\lbrace h_W,h_X \rbrace$, and regularization, $\lbrace \alpha,\rho \rbrace$, parameters. For the choice of the tuning parameters $\lbrace h_W,h_X,\alpha \rbrace$, we follow \cite{centorrino2013} and use a sequential leave-one-out cross-validation approach. We instead select the regularization parameter $\rho$, for $j = \lbrace 1,2 \rbrace$ as

\begin{equation} \label{eq:mincritrho}
\rho^\ast_j = {{\rm arg}\min}_{\rho} \Vert \hat\varphi_{\alpha,\rho,j} - \hat\varphi_{\alpha} \Vert^2 + \varsigma \Vert F_{j,n} \hat\varphi_{\alpha,\rho,j} \Vert^2,
\end{equation}

with $\varsigma > 0$ a constant. The first term of this criterion function is a statistical loss that we incur into when we impose the fairness constraint. The second term instead represents the distance of our estimator to full fairness. The smaller the norm of the second term, the closer we are to obtain a \textit{fair} estimator. For instance, if our unconstrained estimator, $\hat\varphi_{\alpha}$ is \textit{fair}, then the second term will be identically zero for any value of $\rho$, while the first term will be zero for $\rho = 0$, and then would increase as $\rho \rightarrow \infty$. The constant $\varsigma$ serves as a subjective weight for fairness. In principle, one could set $\varsigma = 1$. Values of $\varsigma$ higher than $1$ imply that the decision maker considers deviations from fairness to be very costly and thus prefers them to be penalized more heavily. The opposite is true for values of $\varsigma < 1$. 

\section{An illustration} \label{sec:illustration}

We consider the following illustration of the model described in the previous section. We generate a random vector $\tau = (\tau_1,\tau_2)^\prime$ from a bivariate normal distribution with mean $(0,0.5)^\prime$ and covariance matrix equal to
\[
\Sigma_\tau = \begin{bmatrix} 1 & 2\sin(\pi/12) \\ 2\sin(\pi/12) & 1 \end{bmatrix}.
\]

Then, we fix
\begin{align*}
W =& -1 + 2 \Phi(\tau_1) \\
S =& B(\Phi(\tau_2)),
\end{align*}
where $B(\cdot)$ is a Bernoulli distribution with probability parameter equal to $\Phi(\tau_2)$, and $\Phi$ is the cdf of a standard normal distribution.

We then let $\eta$ and $U$ to be independent normal random variables with mean $0$ and variances equal to $0.16$ and $0.25$, respectively, and we generate
\begin{align*}
Z = -1 + 2\Phi\left( W - 0.5S - 0.5WS + 0.5U + \eta \right),
\end{align*}
and
\[
Y = \varphi_0(Z)+ \varphi_1(Z) S + U,
\]
where $\varphi_0(Z) = 3Z^2$, and
\[
\varphi_1(Z) = 1 - 5Z^3.
\]
In this illustration, the random variable $Z$ can be thought to be an observable characteristic of the individual, while $S$ could be a sensitive attribute related, for instance, to gender, or ethnicity. Notice that the true regression function is \textit{not fair} in the sense of either Definition \ref{def:fairness1} or Definition \ref{def:fairness2}. This reflects the fact that real data may contain a bias with respect to the sensitive attribute, which is often the case in practice. We fix the sample size at $n= 1000$, and we use Epanechnikov kernels for estimation.

\begin{figure}[!h]
\centering
\includegraphics[scale=.5]{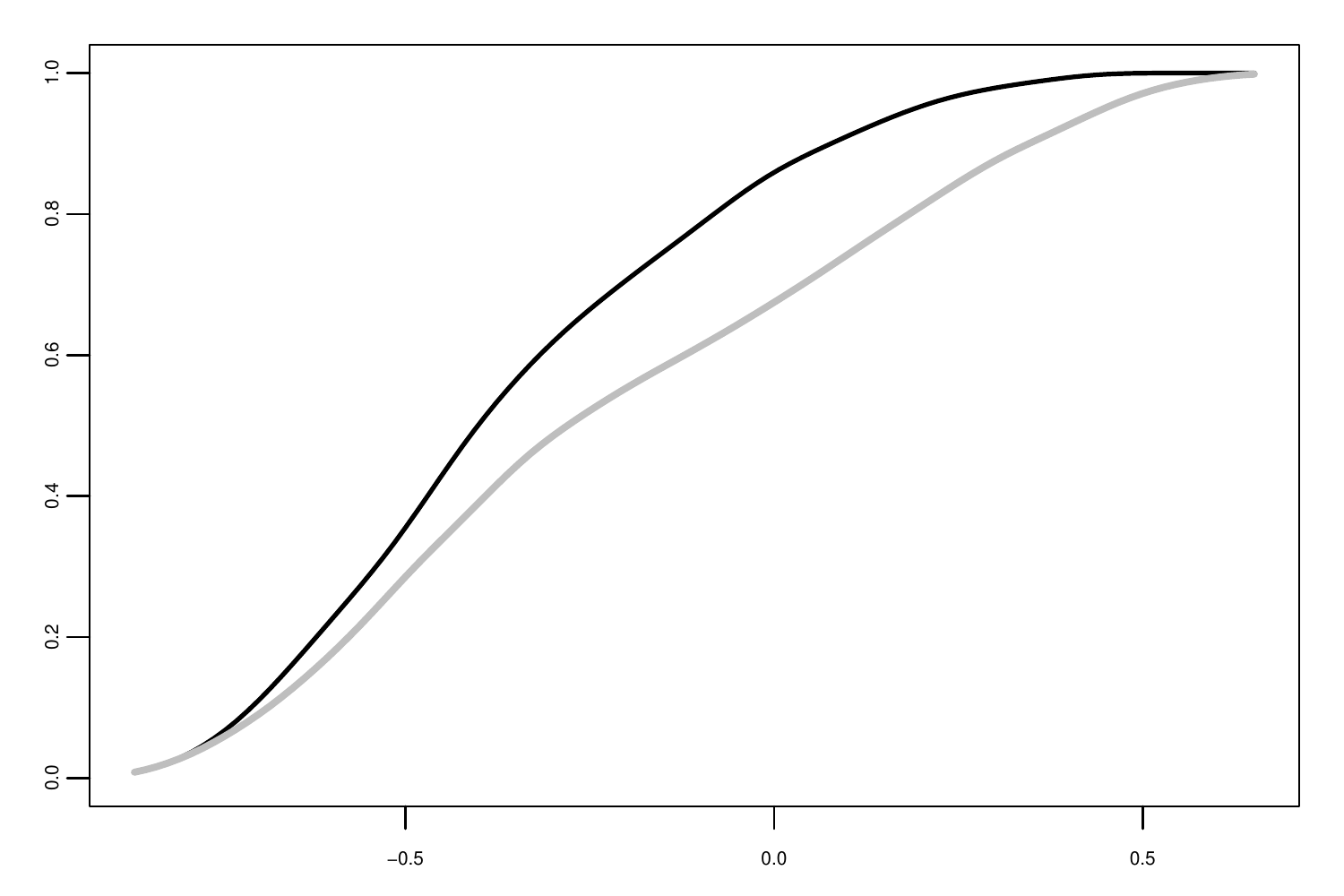}
\caption{Empirical CDF of the endogenous regressor $Z$, conditional of the sensitive attribute $S$. CDF of $Z \vert S = 0$, solid grey line; CDF of $Z \vert S = 1$, solid black line.}
\label{fig:fct_densz}
\end{figure}

In Figure \ref{fig:fct_densz}, we plot the empirical cumulative distribution function (CDF) of $Z$ given $S=0$ (solid grey line), and of $Z$ given $S=1$ (solid black line). We can see that the latter stochastically dominates the former. This can be interpreted as the fact that systematic differences in group's characteristics that can generate systematic differences in the outcome, $Y$, even when the sensitive attribute $S$ is not directly taken into account.

We compare the unconstrained estimator, $\hat\varphi_\alpha$, with the \textit{fairness-constrained} estimators in the sense of Definitions \ref{def:fairness1} and \ref{def:fairness2}.

\begin{figure}[!h]
\centering
\subfigure[$\varphi_0(x) = 3x^2$]{
\includegraphics[scale=.35]{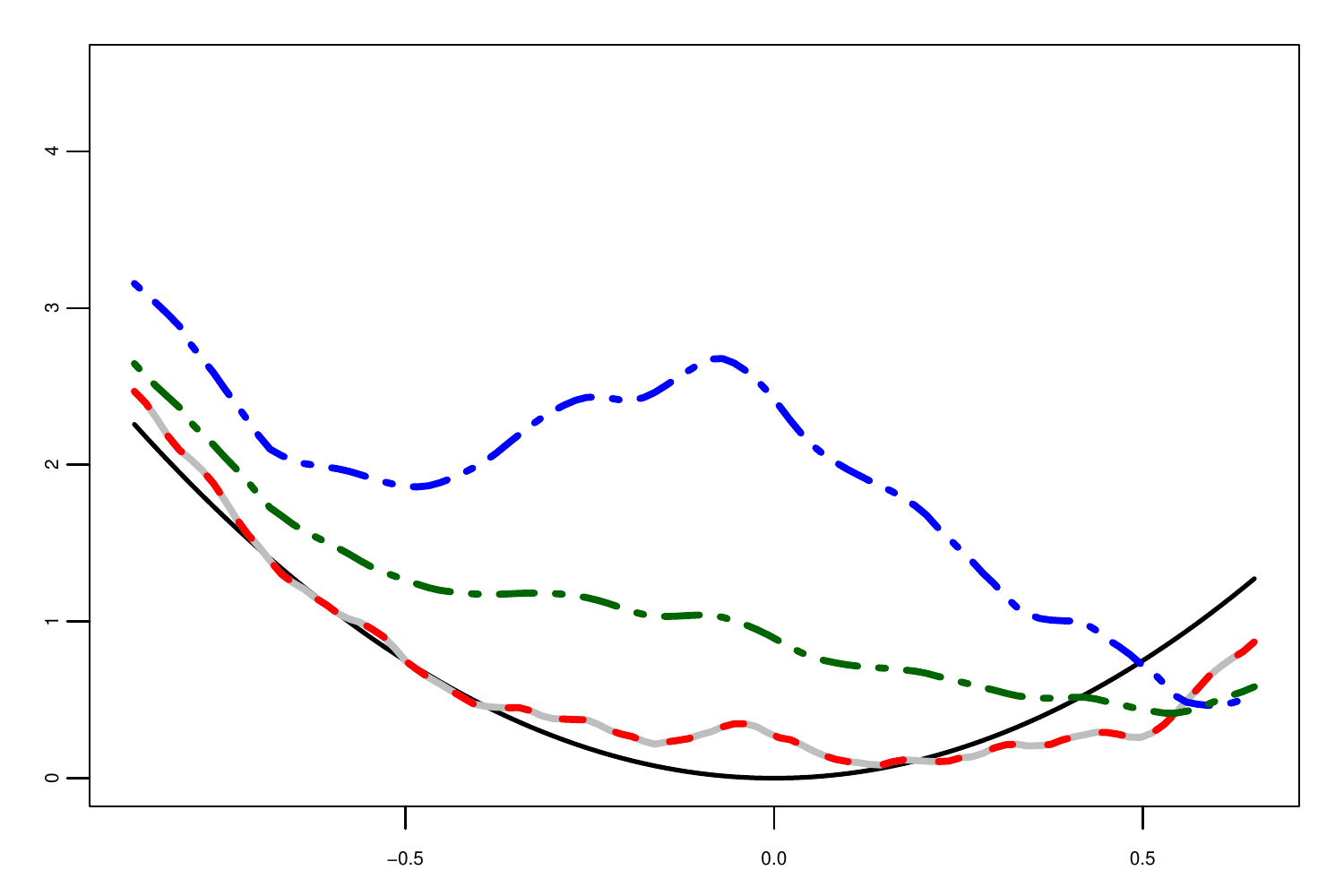}
}
\subfigure[$\varphi_1(x) = 1 - 5x^3$]{
\includegraphics[scale=.35]{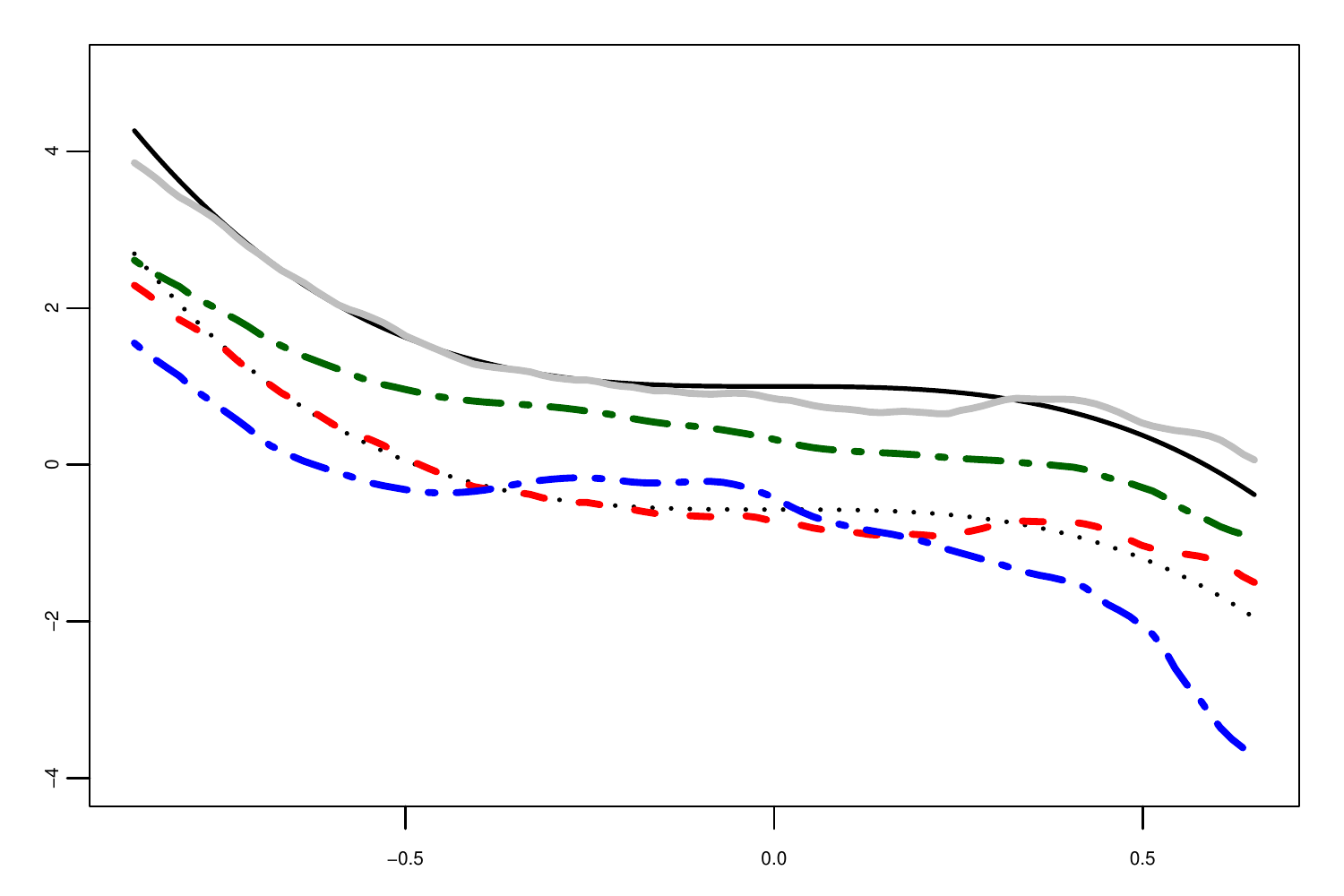}
}
\caption{Estimation using the definition of fairness in \ref{def:fairness1}. Black line, true function; dotted black line, true function with fairness constraint; gray line, $\hat\varphi_\alpha$; dashed red line, $\hat\varphi_{\alpha,F}$; dash blue line, $\hat\varphi_{\alpha,K_F}$, ; dash green line, $\hat\varphi_{\alpha,\rho}$.}
\label{fig:fct_sim1}
\end{figure}

\begin{figure}[!h]
\centering
\subfigure[$\varphi_0(x) = 3x^2$]{
\includegraphics[scale=.36]{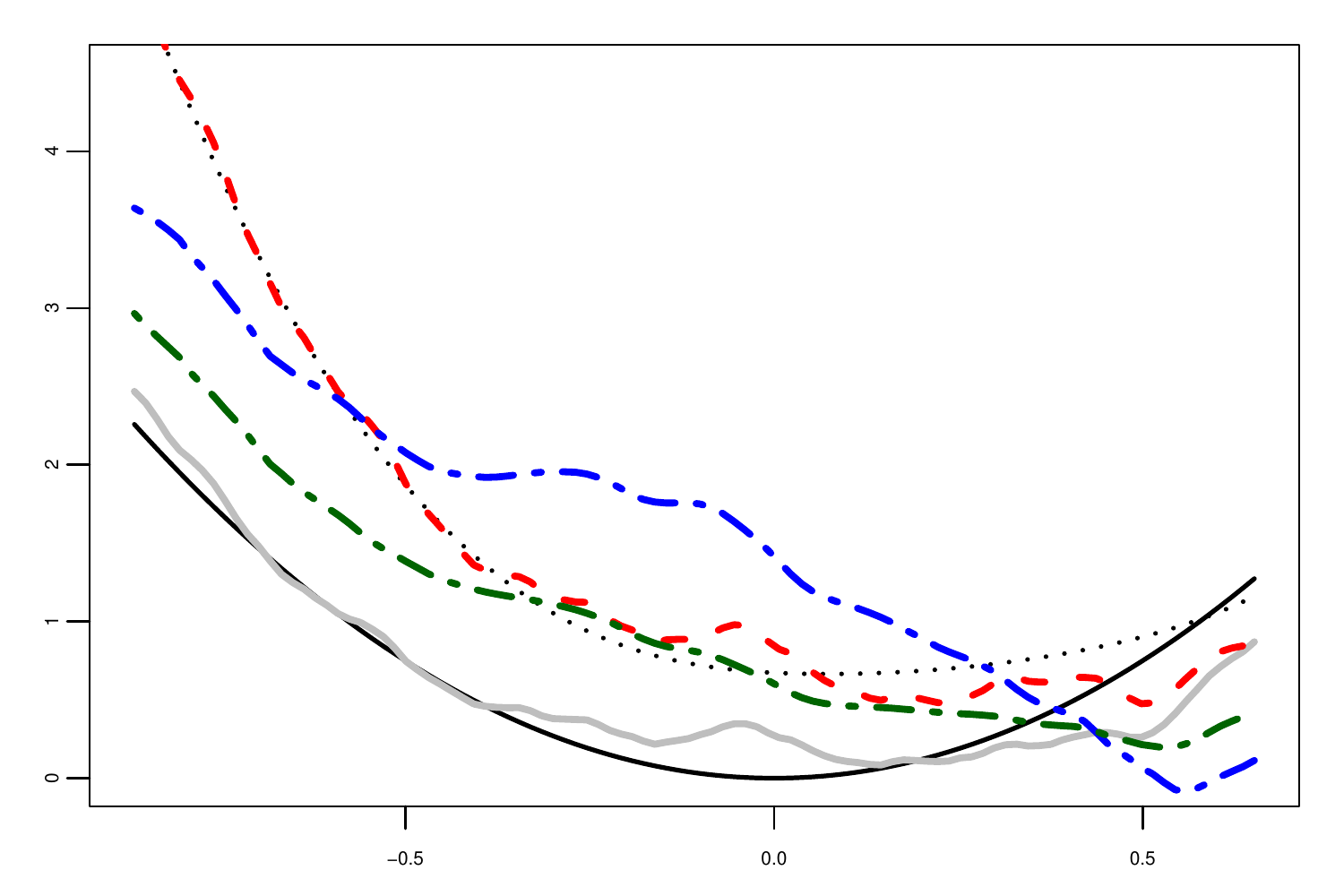}
}
\subfigure[$\varphi_1(x) = 1 - 5x^3$]{
\includegraphics[scale=.36]{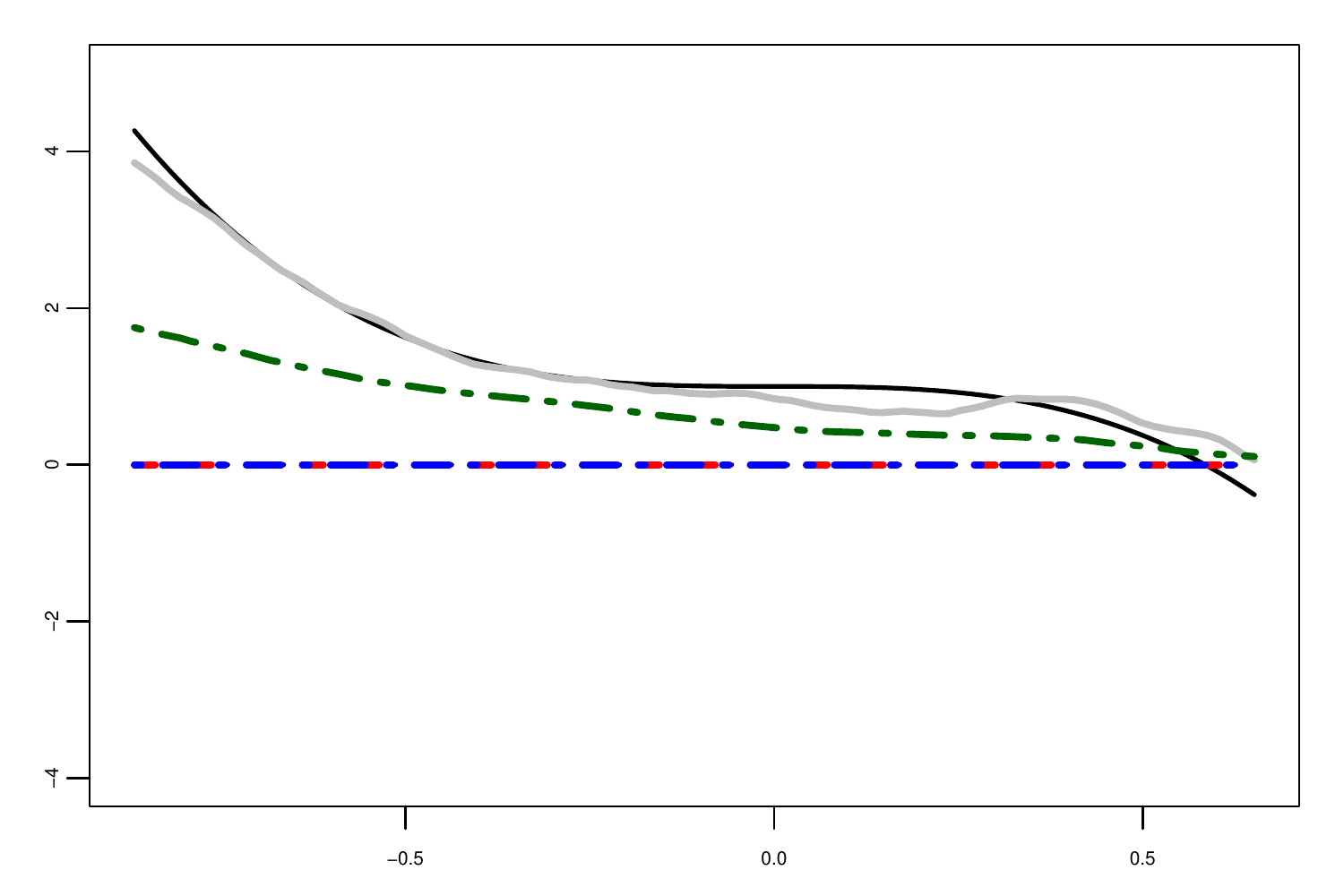}
}
\caption{Estimation using the definition of fairness in \ref{def:fairness2}. Black line, true function; dotted black line, true function with fairness constraint; gray line, $\hat\varphi_\alpha$; dashed red line, $\hat\varphi_{\alpha,F}$; dash blue line, $\hat\varphi_{\alpha,K_F}$; dash green line, $\hat\varphi_{\alpha,\rho}$.}
\label{fig:fct_sim2}
\end{figure}

In Figures \ref{fig:fct_sim1} and \ref{fig:fct_sim2}, we plot the estimators of the functions $\lbrace \varphi_0,\varphi_1 \rbrace$, under the fairness constraints in Definitions \ref{def:fairness1} and \ref{def:fairness2}, respectively. Notice that, as expected, the estimator which imposes approximate fairness through the penalization parameter $\rho$ lays somewhere in between the unconstrained estimator, and the estimators which impose full fairness. 

\begin{figure}[!h]
\centering
\subfigure[Definition \ref{def:fairness1}]{
\includegraphics[scale=.36]{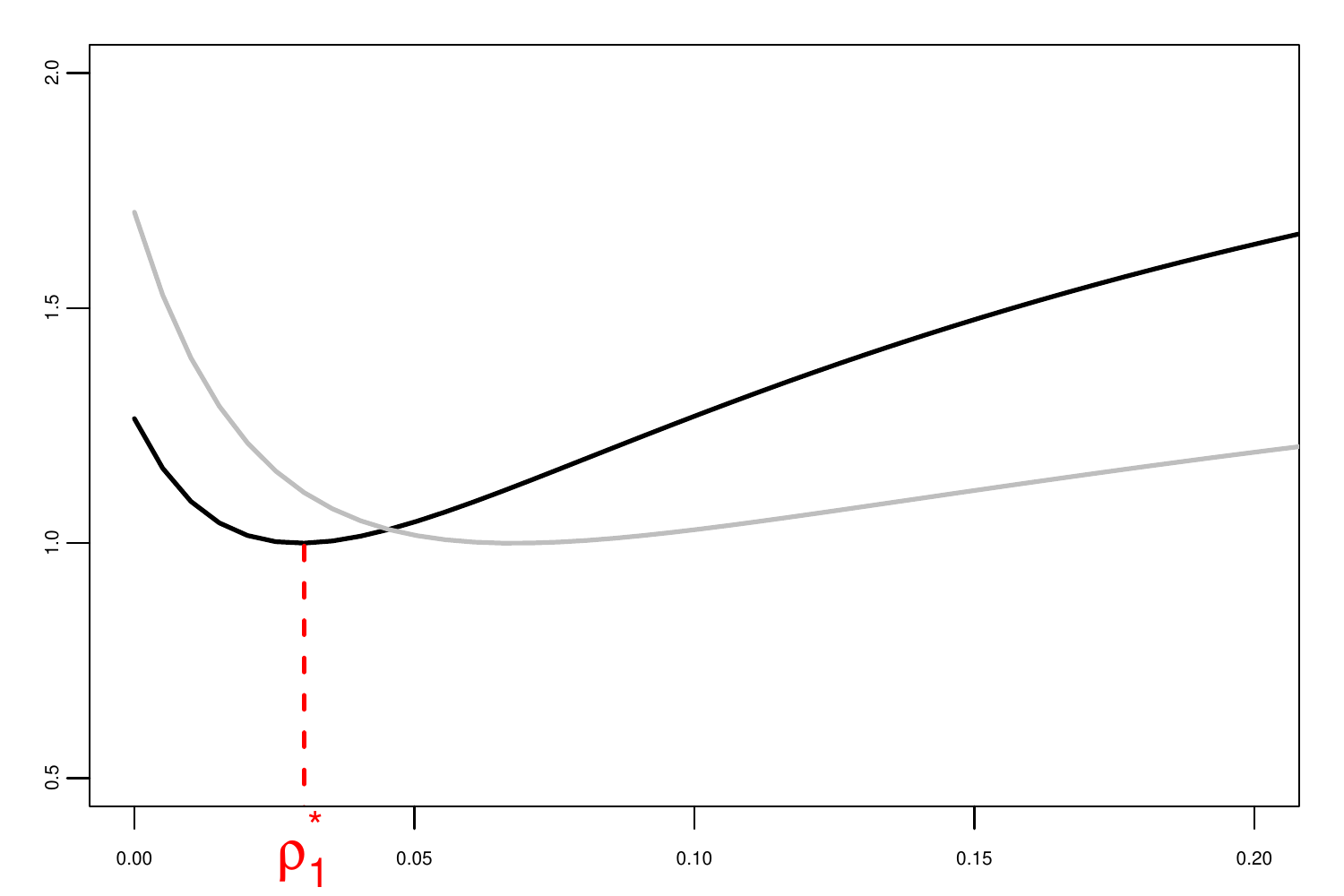}
}
\subfigure[Definition \ref{def:fairness2}]{
\includegraphics[scale=.36]{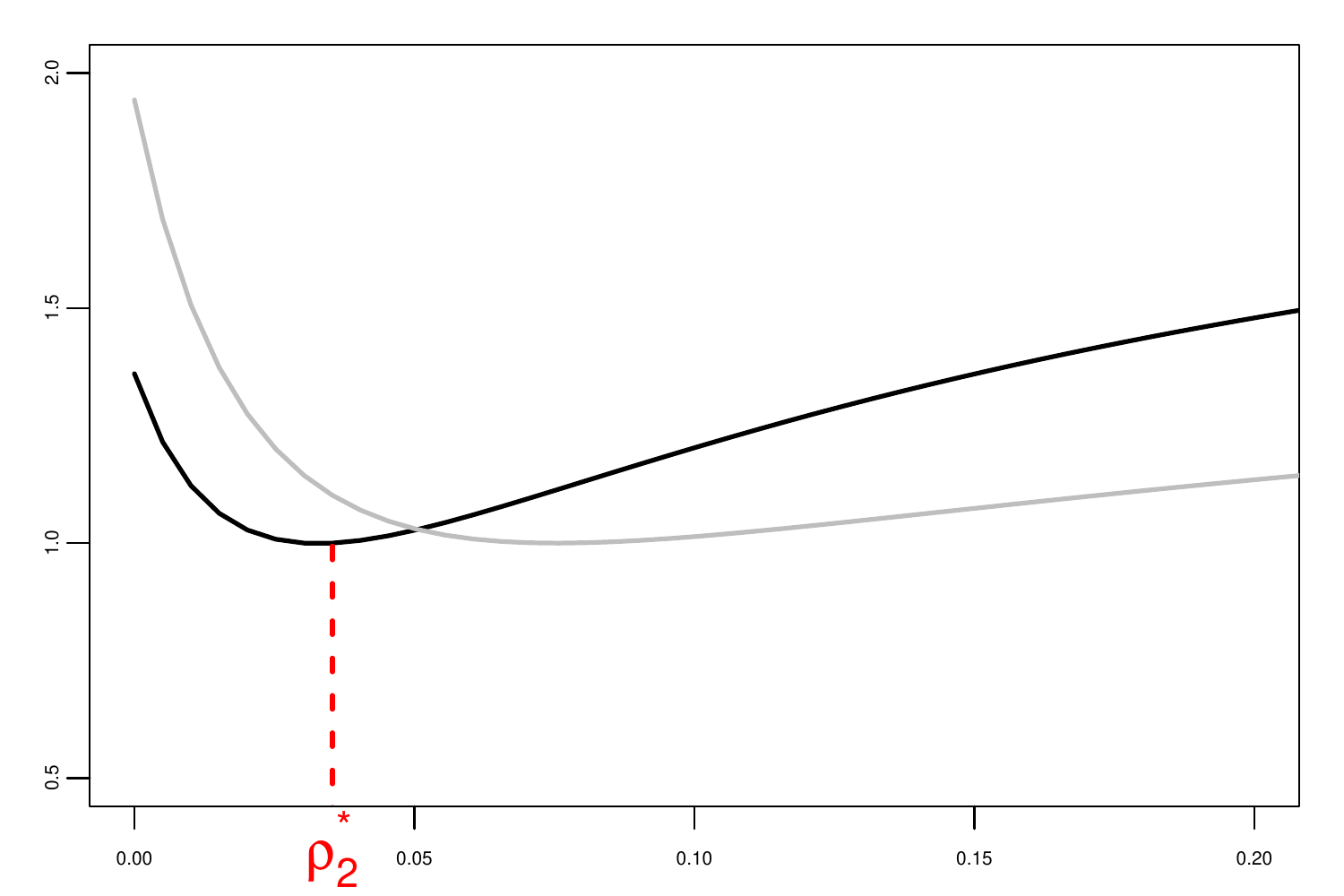}
}
\caption{Choice of the optimal value of $\rho$. }
\label{fig:fct_rhosel}
\end{figure}

In Figure \ref{fig:fct_rhosel}, we depict the objective function in equation \eqref{eq:mincritrho} for the optimal choice of $\rho$, using both Definition \ref{def:fairness1} (left panel) and Definition \ref{def:fairness2} (right panel). The optimal value of $\rho$ is obtained in our case by fixing $\varsigma = 1$ (solid black line). However, if a decision maker wished to impose more \text{fairness}, this could be achieved by setting $\varsigma > 1$. For illustrative purposes, we also report the objective function when $\varsigma = 2$ (solid grey line). It can be seen that this leads to a larger value of $\rho^\ast$, but also that the objective function tends to flatten out. 

\begin{figure}[!h]
\centering
\subfigure[Definition \ref{def:fairness1}]{
\includegraphics[scale=.36]{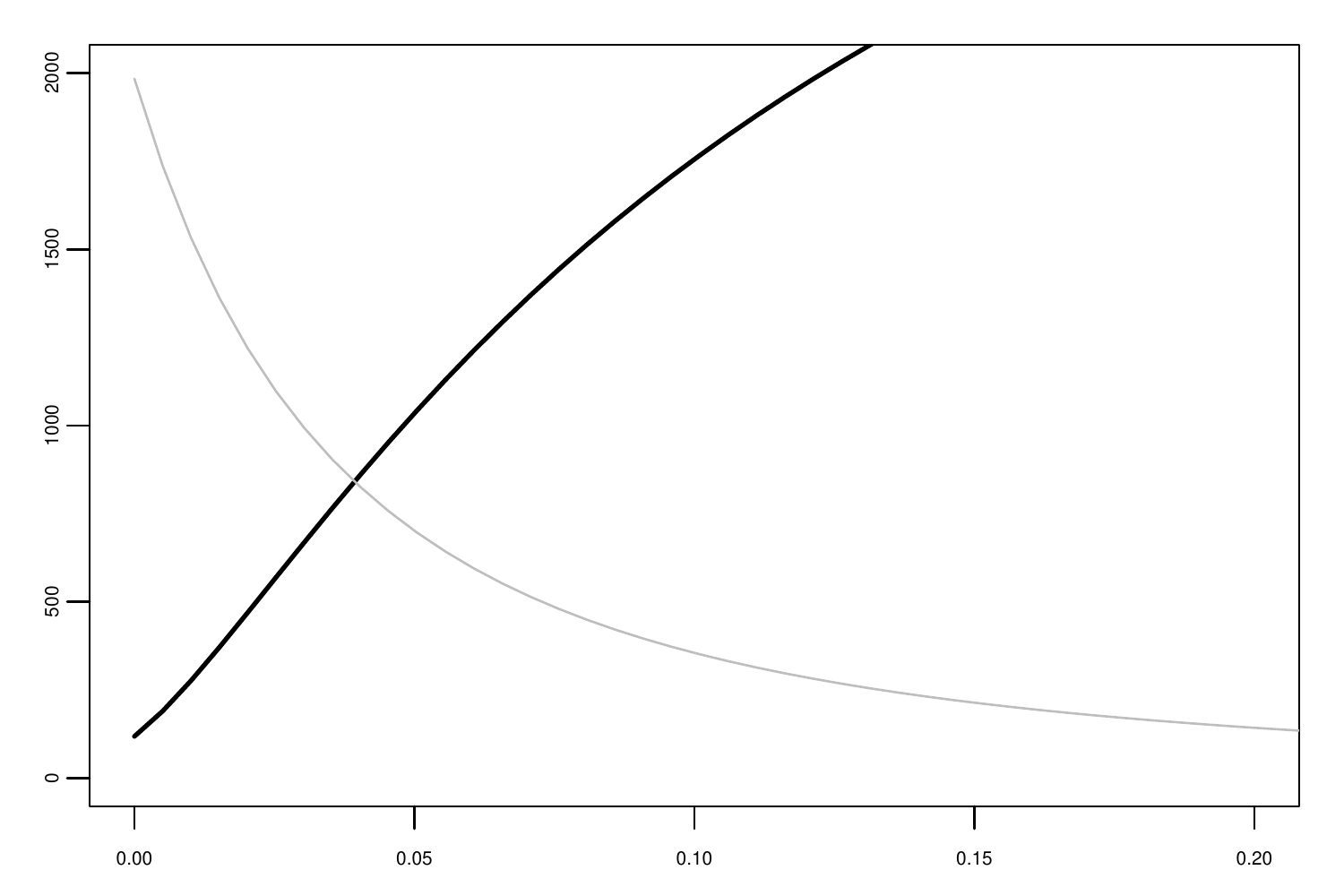}
}
\subfigure[Definition \ref{def:fairness2}]{
\includegraphics[scale=.36]{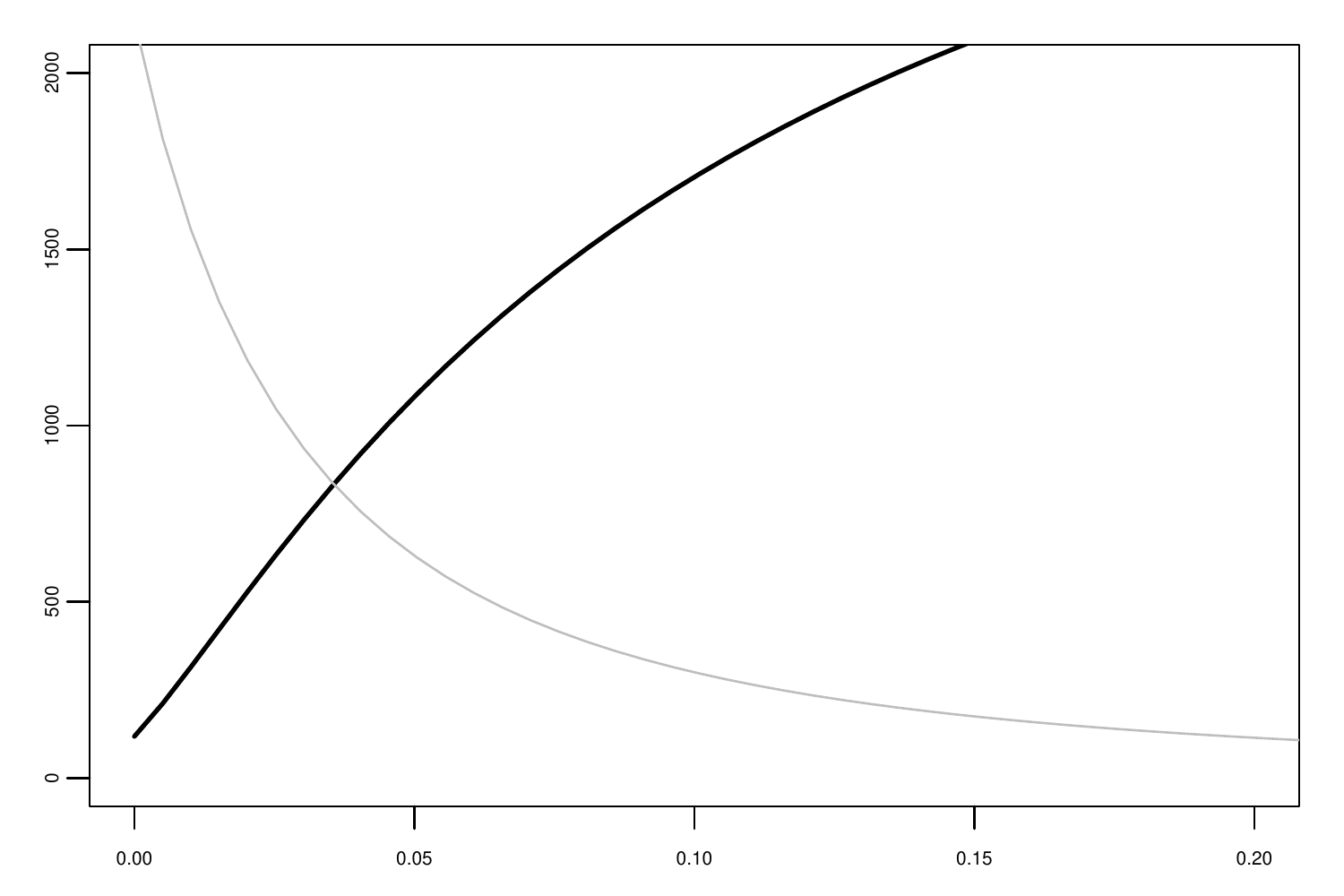}
}
\caption{Cost and benefit of fairness as a function of the penalization parameter $\rho$.}
\label{fig:fct_rhofun}
\end{figure}

We also present in Figure \ref{fig:rhofun} the trade-off between the statistical loss (solid black line), $\Vert \hat\varphi_{\alpha,\rho,j} - \hat{\varphi}_\alpha \Vert^2$, which can be interpreted as the cost of imposing a fair solution, and the benefit of fairness (solid grey line), which is measured by the squared norm of $F_{n,j} \hat\varphi_{\alpha,\rho,j}$, when $j = \lbrace 1,2 \rbrace$ to reflect both Definitions \ref{def:fairness1} (left panel) and \ref{def:fairness2} (right panel). In both cases, we fix $\varsigma = 1$. The upward sloping line is the squared deviation from the unconstrained estimator which increases with $\rho$. The downward sloping curve is the norm of the projection of the estimator onto the space of fair functions, which converges to zero as $\rho$ increases.  

\begin{figure}[!h]
\centering
\subfigure[Definition \ref{def:fairness1}]{
\includegraphics[scale=.36]{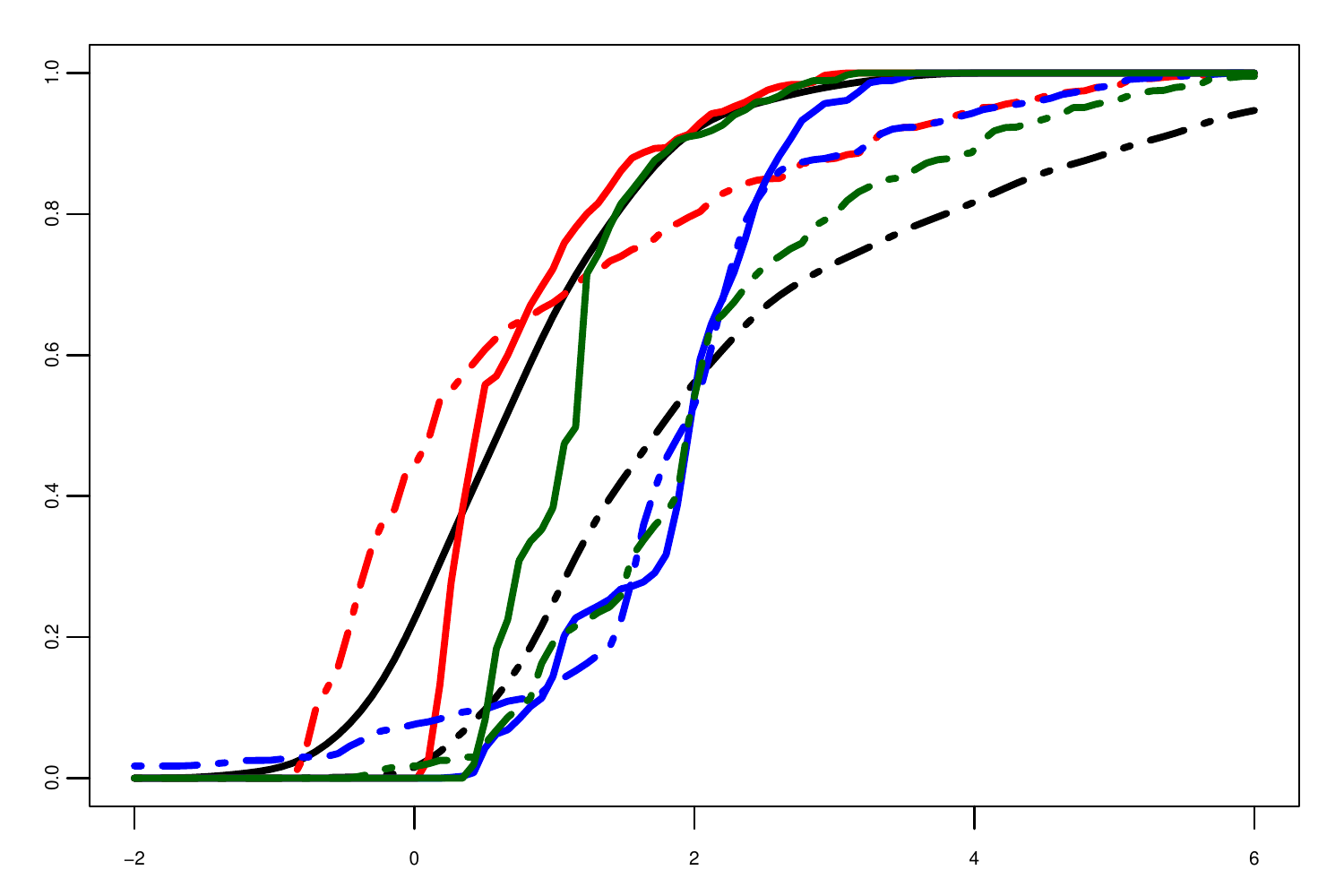}
}
\subfigure[Definition \ref{def:fairness2}]{
\includegraphics[scale=.36]{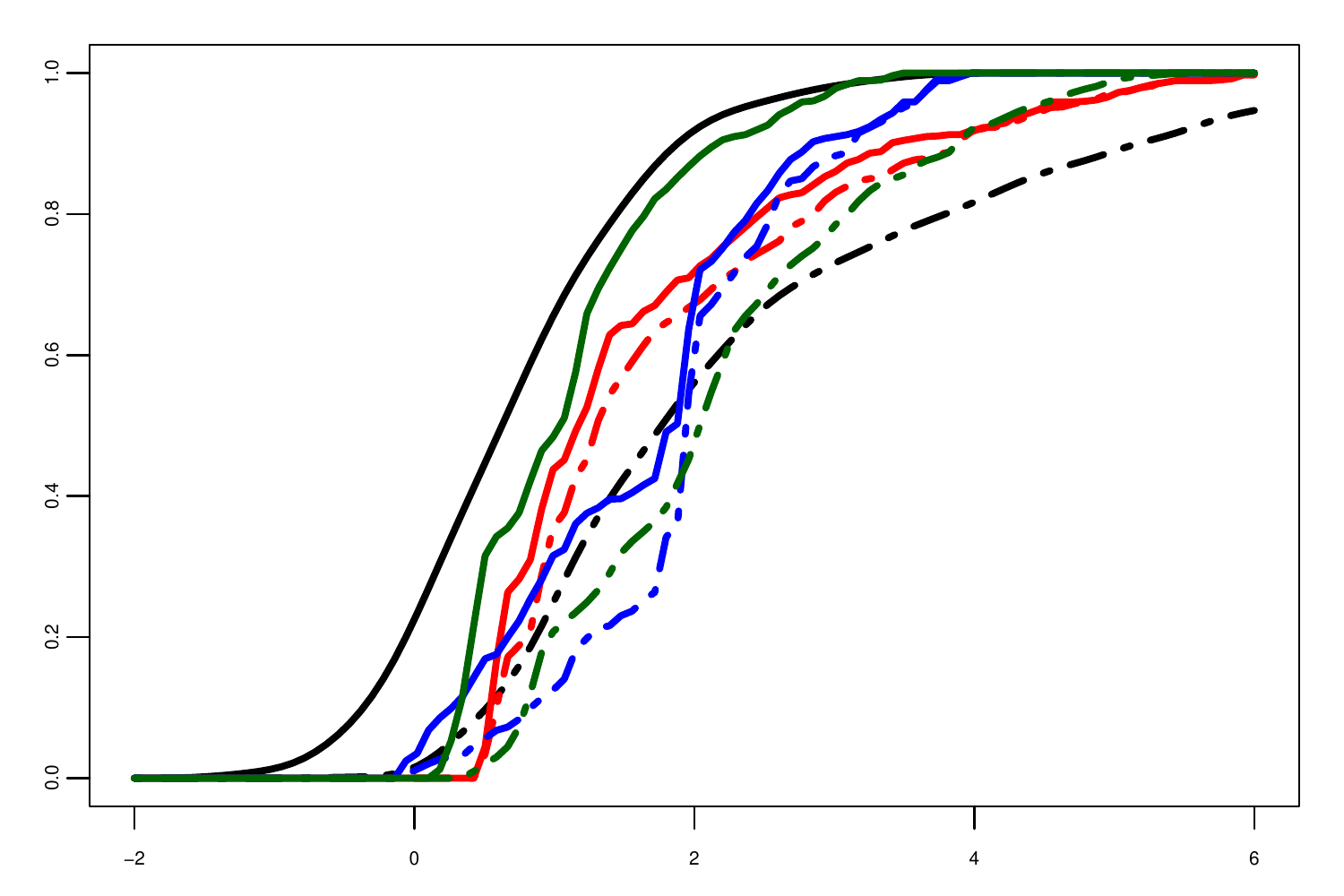}
}
\caption{Density of the predicted values from the constrained models. Solid line is group $S = 0$, and dashed-dotted line is group $S=1$. Black lines are the densities of the observed data; red lines are from constrained model 1; blue from constrained model 2; green from constrained model 3.}
\label{fig:predval}
\end{figure}

Finally, it is interesting to assess how the different definitions of fairness and the different implementations affect the distribution of the predicted values. This prediction is done \textit{in-sample} as its goal is not to assess the predictive properties of our estimator but rather to assess how the different definitions of fairness and the various ways to impose the fairness constraint in estimation affect the distribution of the model predicted values. 

The black lines in Figure \ref{fig:predval} represent the empirical CDF of the dependent variable $Y$ for $S =0$ (dashed-dotted black line), and $S=1$ (dashed black line). This is compared with the predictions using estimators 1 (red lines), 2 (blue lines), and 3 (green lines). In the data, the distribution of $Y$ given $S=1$ stochastically dominates the distribution of $Y$ given $S = 0$. 

Notice that in case of fairness as defined in \ref{def:fairness1}, the estimator which modifies the conditional expectation operator to projects directly onto the space of fair functions seems to behave best in terms of fairness, as the distribution of the predicted values for groups $0$ and $1$ are very similar. The estimator which imposes approximate fairness obviously lies somewhere in between the data and the previous estimator. The projection of the unconstrained estimator onto the space of fair functions does not seem to deliver an appropriate distribution of the predicted values. What happens is that this estimator penalizes people in group $1$ with low values of $Z$, in order to maintain fairness on the average while maintaining a substantial difference in the distribution of the two groups. 

Differently, in the case of fairness as defined in \ref{def:fairness2}, the projection of the unconstrained estimator seems to behave best. However, this may be due to the fact that the distribution of $Z$ given $S=0$ and $S=1$ are substantially similar. If however, there is a more difference in the observable characteristics by group, this estimator may not behave as intended. 

\section{Conclusions}

In this chapter, we consider the issue of estimating a structural econometrics model when a fairness constraint is imposed on the solution. We focus our attention on models when the function is the solution to a linear inverse problem, and the fairness constraint is imposed on the included covariates and can be expressed as a linear restriction on the function of interest. We also discuss how to impose an \textit{approximately fair} solution to a linear functional equation and how this notion can be implemented to balance accurate predictions with the benefits of a fair machine learning algorithm. We further present regularity conditions under which the fair approximation converges towards the projection of the true function onto the null space of the fairness operator. Our leading example is a nonparametric instrumental variable model, in which the fairness constraint is imposed.  We detail the example of such a model when the sensitive attribute is binary and exogenous \cite{centorrino2017}. 

The framework introduced in this chapter can be extended in several directions. The first significant extension would be to consider models in which the function $\varphi_\dagger$ is the solution to a nonlinear equation. The latter can arise, for instance, when the conditional mean independence restriction is replaced with full independence between the instrumental variable and the structural error term \cite{centorrinofeveflorens2019,centorrinoflorens2021}. Moreover, one can potentially place fairness restrictions directly on the decision algorithm or on the distribution of predicted values. These restrictions usually imply that the fairness constraint is nonlinear, and a different identification and estimation approach should be employed. \\ In this work, we restrict ourselves to group fairness notions and did not consider fairness at an individual notions such as in \cite{kusner2017counterfactual},  \cite{de2021transport},  which could enable to understand fairness in econometry from a causal point of view.\\
 Finally, the fairness constraint imposed in this paper is limited to the function, $\varphi$. However, other constraints may be imposed directly on the functional equation. For instance, on the selection of the instrumental variables, which will be the topic of a further work. 

\bibliographystyle{unsrt}  
\bibliography{labib,fairness}
 
 \end{document}